\documentclass[9pt]{extarticle}
\usepackage[utf8]{inputenc}
\usepackage{mathtools}
\usepackage[export]{adjustbox}
\usepackage{amssymb}
\usepackage{braket}
\usepackage{bbm}
\usepackage[mathcal]{euscript}
\usepackage[a4paper, total={6.25in, 8in}]{geometry}
\usepackage[title]{appendix}
\usepackage{amsmath}
\usepackage{amsthm}
\usepackage{verbatim}
\usepackage{marvosym}
\usepackage{caption}
\usepackage{subcaption}
\usepackage{relsize}
\usepackage{epigraph}
\newtheorem{lemma}{Lemma}[section]
\newtheorem{definition}{Definition}[section]
\newtheorem{corollary}{Corollary}[section]
\newtheorem*{nonumthm}{Theorem}
\AtBeginDocument{}

\usepackage{etoolbox}
    \makeatletter
    \newlength\epitextskip
    \pretocmd{\@epitext}{\em}{}{}
    \apptocmd{\@epitext}{\em}{}{}
    \patchcmd{\epigraph}{\@epitext{#1}\\}{\@epitext{#1}\\[\epitextskip]}{}{}
    \makeatother

\usepackage{fourier}
\usepackage{hyperref}
\hypersetup{
    colorlinks=true,
    linkcolor=blue,
    filecolor=magenta,      
    urlcolor=cyan,
}
\title{On the relativistic quantum mechanics of\\ a photon between two electrons\\ in 1+1 dimensions}
\author{Lawrence Frolov$^*$, Samuel Leigh$^\dagger$, and Shadi Tahvildar-Zadeh$^*$\\[10pt]
{\small $^*$ Department of Mathematics, Rutgers University (New Brunswick)}\\
{\small $^\dagger$ Department of Physics, University of Washington}
}
\date{November 2024}
\numberwithin{equation}{section}
\newtheorem{theorem}{Theorem}[section]
\newtheorem{proposition}{Proposition}[section]
\theoremstyle{definition}
\newtheorem{remark}{Remark}[section]
\begin{document}
\maketitle
\begin{abstract}
A Lorentz-covariant system of wave equations is formulated for a quantum-mechanical three-body  system  in  one  space  dimension,  comprised  of  one photon and two identical massive spin one-half Dirac particles, which can be thought of as two electrons (or alternatively, two positrons). Manifest covariance is achieved using Dirac's formalism of multi-time wave functions, i.e, wave functions $\Psi(\textbf{x}_{\text{ph}},\textbf{x}_{\text{e}_1},\textbf{x}_{\text{e}_2})$ where $\textbf{x}_{\text{ph}},\textbf{x}_{\text{e}_1},\textbf{x}_{\text{e}_2}$ are generic spacetime events of the photon and two electrons respectively. Their  interaction  is  implemented  via  a  Lorentz-invariant  no-crossing-of-paths  boundary  condition  at  the  coincidence  submanifolds $\{\textbf{x}_{\text{ph}}=\textbf{x}_{\text{e}_1}\}$ and $\{\textbf{x}_{\text{ph}}=\textbf{x}_{\text{e}_2}\}$ compatible with conservation of probability current. The corresponding initial-boundary value problem is shown to be well-posed, %under the additional assumption of anti-symmetry given by the Pauli exclusion principle, and a closed-form solution to the ensuing coupled system of Klein-Gordon and transport equations is given. 
and it is shown that the unique solution can be represented by a convergent infinite sum of Feynman-like diagrams, each one corresponding to the photon bouncing between the two electrons a fixed number of times.
\end{abstract}
\section{Introduction}
\epigraph{\centering The Compton effect, at its discovery, was regarded as a simple collision of two bodies, and yet the  detailed  discussion  at  the  present  time  involves  the  idea  of  the  annihilation  of  one  photon and the simultaneous creation of one among an infinity of other possible ones.  We would like to be able to treat the effect as a two-body problem, with the scattered photon regarded as the same individual as the incident, in just the way we treat the collisions of electrons.}{\textbf{C. G. Darwin,} \textit{Notes on the Theory of Radiation (1932)}}
Despite it being over 90 years since the words above appeared in print, it is still true that "the detailed discussion [of the Compton effect] at the present time involves the idea of annihilation of one photon and the simultaneous creation of one among an infinity of others;" see \cite{Comp1,Comp2,Comp3,Comp4,Comp5}. Not only has Darwin’s goal, of treating electrons and photons on an equal footing
within a quantum-mechanical framework of a fixed number of particles, remained elusive, but modern quantum-field theorists have even come to the conclusion that “a relativistic quantum
theory of a fixed number of particles, is an impossibility;” \cite{Wein}.
\par
The impossible was however  done in \cite{KLTZ} in the $1+1$ dimensional setting, where the authors were able to successfully couple the relativistic quantum-mechanical photon wave equation of \cite{KTZ} with Dirac's relativistic quantum-mechanical electron wave equation in a Lorentz-covariant manner to accomplish
Darwin’s goal: to “treat the [Compton] effect as a two-body problem, with the scattered photon
regarded as the same individual as the incident, in just the way we treat the collisions of electrons”. They did not do this by studying a quantization of the Maxwell-Dirac system, since Maxwell's equations in $1+1$ dimensions do not describe any electromagnetic radiation that can be quantized. Instead, they worked in Dirac's manifestly Lorentz covariant formalism of multi-time wave functions \cite{Dir32}, with the Compton interaction induced on the two body wave function $\Psi^{(2)}(\textbf{x}_\text{ph},\textbf{x}_\text{e})$ by imposing a suitable boundary condition on it at the subset of co-incident events $\{\textbf{x}_\text{e}=\textbf{x}_\text{ph}\}$. 
\par
In this paper we see how, after reformulating their results in terms of Feynman-like diagrams\footnote{As we will see, the main difference between these and the traditional Feynman diagrams is that the vertices in our diagrams have degree 4, not 3.}, that particular boundary condition induces a local pair interaction between the electron and the photon that is worthy of being called "Compton Scattering in $1+1$ dimensions". Of course, a $1+1$ dimensional treatment
of the Compton effect is too simplistic to allow a comparison with empirical data, so our analysis of this toy model should only be seen as a \textit{proof of concept}. We will see that after enough time has passed for the photon to interact with the electron, the two-body wave function becomes a sum of two types of diagrams: diagrams where the electron and photon freely propagate, and diagrams where the pair become incident before freely propagating:
\begin{figure}[htbp]
\centering
$\Psi^{(2)}(\textbf{x}_\text{ph},\textbf{x}_\text{e})$ \qquad $=$ \qquad
\includegraphics[width=1.8cm, height=1cm, valign=c]{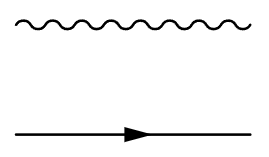}
        \quad $\mathlarger{+} \quad $  \includegraphics[width=1.8cm, height=1cm, valign=c]{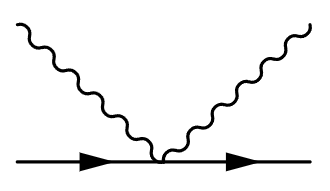}
        \caption{The wave function can be represented as a sum of freely propagating data and data resulting from a collision.}
\end{figure}
\par
The main purpose of this paper is to extend the results of \cite{KLTZ} to the case of two electrons interacting solely through their individual contact interactions with a mediating photon. We expect that placing a photon between the two freely propagating electrons will produce an effective interaction along the electron's light cones as the photon bounces back and forth between them. However, allowing the photon to bounce arbitrarily between the two massive electrons has a serious consequence. So long as the electrons remain sufficiently close together throughout their contact interactions, any number of bounces can occur in a finite time. Since our wave function must be a super-position over all possible interactions, it becomes an infinite sum of diagrams of the following form\footnote{Note that in this and future diagram sums like it, one diagram represents also the ones symmetric to it under the exchange of label for the two electrons.}
\begin{figure}[htbp]
\centering
$\Psi^{(3)}(\textbf{x}_\text{ph},\textbf{x}_{\text{e}_1},\textbf{x}_{\text{e}_2})$ \qquad $=$ \qquad
\includegraphics[width=1.8cm, height=1cm, valign=c]{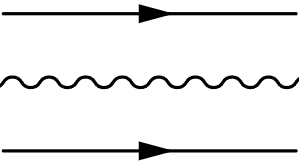} \quad
        $\mathlarger{+} \quad $
        \includegraphics[width=1.78cm, height=1cm, valign=c]{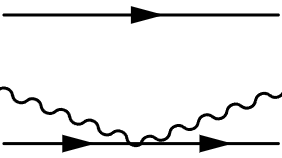}
        \quad 
        $\mathlarger{+} \quad  $
        \includegraphics[width=1.8cm, height=0.9cm, valign=c]{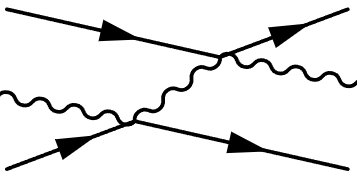}
        \quad
        \\[10pt] $\mathlarger{+}   $
% \end{figure}
% \begin{figure}[htbp]
% \centering
\includegraphics[width=1.8cm, height=0.9cm, valign=c]{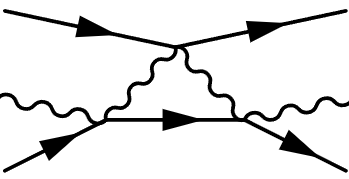}
        $\mathlarger{+} \quad $
\includegraphics[width=1.8cm, height=0.9cm, valign=c]{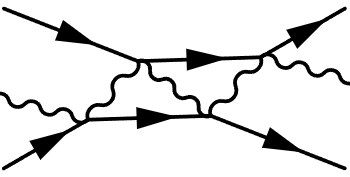} \quad
        $\mathlarger{+} \quad  $
        \includegraphics[width=1.8cm, height=0.9cm, valign=c]{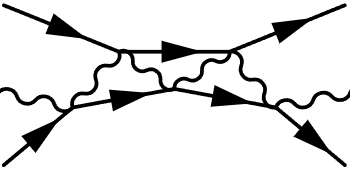} \quad
        $\mathlarger{+} \quad  $
        \includegraphics[width=1.8cm, height=0.9cm, valign=c]{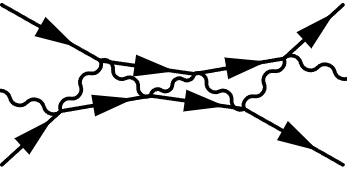}
        $\mathlarger{+} ... $
        \caption{\label{fig:Fey3} The photon could have bounced arbitrarily many times between the electrons before the three propagated to their final destinations.}
\end{figure}
\par
This makes it effectively impossible to write down a closed-form solution for the three-body wave function evolution, in contrast to what was done for the electron-photon wave function in \cite{KLTZ}. Worse, it is not clear that this infinite sum converges in some appropriate function space. %Physically, one may wish to argue that a large number of bounces in finite time is improbable, and so the contribution of the higher order diagrams should become smaller and smaller. However, these probabilities are not zero, so we cannot immediately rule out the divergence of this sum.

Our goal here is to take one more step in the direction of a fully rigorous and manifestly Lorentz-covariant quantum theory that incorporates creation, annihilation, and scattering of any number of photons and electrons in one space dimension.  The following is an informal statement of the main result of this paper (for the precise formulation, see Theorem~\ref{thm:main}.)

\bigskip
\noindent{\bf Main Result:} 
   {\em Let $\Psi_\epsilon$ denote the truncated sum where one discards all diagrams in which the photon bounces off one electron and then the other in less than $\epsilon$ time. Then the sequence $\{\Psi_{1/n}\}$ converges in $L^2$ as $n\to \infty$, and the dynamics produced by the limit correspond to the infinite sum of Feynman diagrams in Fig.~\ref{fig:Fey3}.}
%\end{theorem}

\begin{remark}
Although there is no $1+1$ dimensional toy QED for us to compare our $1+1$ dimensional relativistic toy electron-photon quantum mechanics with, our approach is deeply related to a quantum field-theoretical formulation. Historically there is a direct line from Dirac's multi-time formalism to the QFT formalism of Schwinger and Tomonaga, see \cite{MultiQFT}. To see this connection, let $\hat{\Psi}$ stand for the Heisenberg field operators, $\ket{0}$ denote the zero-particle state and $\ket{\Psi_0}$ the Heisenberg state vector. Then, for all spacelike configurations $(\textbf{x}_1,...,\textbf{x}_N)$ of any number $N$ of particles, the $N$-time $N$-body wave function is given by (neglecting spin)
\begin{equation}
    \Psi^{(N)}(\textbf{x}_1,...,\textbf{x}_N)=\frac{1}{\sqrt{N!}}\braket{0|\hat{\Psi}(\textbf{x}_1) ... \hat{\Psi}(\textbf{x}_N)|\Psi_0}.
\end{equation}
So, the multi-time picture should be regarded as a covariant particle-position representation of the quantum state.
Furthermore, for the quantum field-theoretical dynamics, the particle creation/annihilation formalism requires working with wave functions that belong to the multi-time analog of a Fock space, represented by the sequence
\begin{equation}
    \Psi=(\Psi^{(0)},\Psi^{(1)},\Psi^{(2)},...)
\end{equation}
where each $\Psi^{(N)}$ is an $N$-particle multi-time wave function as above. The dynamics for the multi-time wave functions then typically mix the different particle numbers in the interaction terms, i.e, the right hand side of the abstract Schr\"odinger equation
\begin{equation}\label{Multi-Time QFT}
    i\partial_{t_k}\Psi^{(N)}=(H_k \Psi)^{(N)}
\end{equation}
generally contains $\Psi^{(M)}$ for $M \neq N$.
This quantum-field formalism is usually applied with all times synchronized, cf. \cite{Thaller}; but recently a multi-time formulation has been achieved for certain quantum field theories (essentially those with local
interactions and finite propagation speed), see \cite{MultiWFQFT}.

Unique solvability of the system of evolution equations (\ref{Multi-Time QFT}) requires imposing suitable \textit{boundary conditions} at the subset of co-incident events $\{ \textbf{x}_j=\textbf{x}_k\}$. We refer to the set of coincidence points as an \textit{intrinsic boundary}, because it is system-intrinsic and not determined by some extraneous constraint (such as a container wall). In particular, the quantum
field-theoretical creation/annihilation formalism can be implemented by formulating suitable intrinsic boundary conditions on the $N$-body multi-time wave functions which involve, in addition to such boundary points, interior points of the $N-1$ body wave functions, see \cite{Hamiltonians,Avoid} for single time wave functions. This has led to a notion of \textit{interior-boundary conditions}. For an extension to multi-time wave functions, see \cite{LiNi2020}. In this framework, our quantum-mechanical approach can then be viewed as a special case of the quantum field-theoretical approach, since we restrict our attention to boundary conditions in which the $3$-body dynamics are disjoint from all other $M$-body wave functions, so that we only have scattering. 
\end{remark}
%We aim to take one more step in the direction of a fully rigorous and manifestly Lorentz-covariant quantum-field theory which incorporates creation, annihilation, and scattering of any number of photons and electrons in one space dimension.

The rest of this paper is structured as follows:
%\newline
Section $2$ collects the basic mathematical ingredients needed to formulate our model.
%\newline
In Section $3$ we define the three-time three-body wave function and derive the interaction-inducing \textit{intrinsic boundary} condition from probability conservation.
%\newline
Section $4$ reviews the results of \cite{KLTZ} regarding the multi-time dynamics of a photon-electron system, and reformulates their results in terms of Feynman diagrams.
%\newline
In Section $5$ we use said Feynman diagram formulation of photon-electron interactions to study the multi-time dynamics for a photon between two electrons, and prove our main result. 
%\newline
In Section $6$ we provide a summary and an outlook on open questions left for future work.
Finally, in the Appendix we recall various well-known solution formulas, and we provide the outline of the existence proof by \cite{LiNi2020} for multi-particle dynamics using the contraction mapping theorem.
\section{Preliminaries} 
For any $d \in \mathbb{N}$, we let $\boldsymbol{\eta} = (\eta_{\mu \nu})$ denote the Minkowski metric on Minkowski spacetime $\mathbb{R}^{1,d}$:
\begin{equation}
\boldsymbol{\eta} = \boldsymbol{\eta}^{-1} = \text{diag(1,-1,...,-1)}.
\end{equation}
For $d=1$, the complexified spacetime algebra $\mathcal{A}$ is isomorphic to the algebra of $2 \times 2$ complex matrices $\mathcal{A}:=Cl_{1,1}(\mathbb{R})\cong M_2(\mathbb{C})$. We take our basis to be $\{ \mathbb{1},\gamma^0,\gamma^1,\gamma^0 \gamma^1\}$ where
\begin{equation}
\mathbb{1} \coloneqq
\begin{pmatrix}
1 & 0 \\
0 & 1
\end{pmatrix}
,
\hspace{0.2cm}
\gamma^{0} \coloneqq
\begin{pmatrix}
0 & 1 \\
1 & 0
\end{pmatrix}
,
\hspace{0.2cm}
\gamma^{1} \coloneqq
\begin{pmatrix}
0 & -1 \\
1 & 0
\end{pmatrix}
\end{equation}
satisfy the Clifford Algebra relations
\begin{equation}
\gamma^{\mu} \gamma^{\nu} + \gamma^{\nu} \gamma^{\mu} = 2 \eta^{\mu \nu} \mathbb{1}.
\end{equation}
\subsection{Electron wave function and Dirac Equation}
For a single electron in 1+1 dimensions, the wave function for an electron is given by a rank-one spinor field on the electron's configuration space-time $\mathbb{R}^{1,1}_\text{e}$, equipped with coordinates $(t_\text{e},s_\text{e})$:
\begin{equation}
\Psi_{\text{e}} =
\begin{pmatrix}
\psi_- \\
\psi_+
\end{pmatrix}.
\end{equation}
As a rank-one spinor field, the components of $\Psi_e$ transform under the action of $O(1,1)$ as
\begin{equation}
    \psi_{\pm} \xrightarrow{\Lambda} e^{\mp a/2} \psi_{\pm}, \quad \psi_{\pm} \xrightarrow{\mathbf{P}}  \psi_{\mp}, \quad \psi_{\pm} \xrightarrow{\mathbf{T}}  \psi_{\mp}^*
\end{equation}
where $\Lambda$ refers to proper Lorentz transformations, $\mathbf{P}$ space reflections, $\mathbf{T}$ time reflections, and $^*$ refers to the complex conjugate number.
\par
The electron wave function satisfies the massive Dirac equation
\begin{equation}
-i\hbar \gamma^{\mu}\frac{\partial}{\partial x^\mu_\text{e}}\Psi_\text{e} +m_\text{e}\Psi_\text{e}=0
\end{equation}
where $m_e$ is the mass of the electron, $\hbar$ is Planck's constant, and we have set $c$, the speed of light in vacuum, equal to one. In this paper, we will primarily concern ourselves with solving initial-value problems for such equations. Given smooth, rapidly decaying initial data $\mathring{\Psi}:\mathbb{R}^1_\text{e}\to \mathbb{C}^2$, the unique solution to the Dirac equation IVP 
\begin{equation}\label{Dirac IVP Intro}
    \left\{\begin{array}{rcl}
        -i\hbar \gamma^\mu \frac{\partial}{\partial x^\mu_\text{e}} \Psi_\text{e} + m_\text{e} \Psi_\text{e} &= &0
        \\
        \Psi_\text{e} \big{|}_{t_\text{e}=0}&=&\mathring{\Psi}_\text{e}
    \end{array}\right.
    \end{equation}
can be found by recalling that each component of $\Psi$ satisfies a Klein-Gordon equation. Solving the Dirac IVP becomes equivalent to solving the Klein-Gordon IVP \begin{equation}\label{1 El KG IVP}
        \left\{\begin{array}{rcl}
            \Box_\text{e} \psi_\pm+\omega^2 \psi_\pm&=&0
            \\
            \psi_\pm \big{|}_{t_\text{e}=0}&=&A_\pm(s_\text{e})
            \\
            \partial_{t_\text{e}} \psi_\pm \big{|}_{t_\text{e}=0}&=&B_\pm(s_\text{e})
        \end{array}\right.
    \end{equation}
    where $\omega=\frac{m_\text{e}}{\hbar}$ and the initial data $B_{\pm}$ is derived from the Dirac equation and will depend on $A_{\mp}$ and $\partial_1 A_{\pm}$.
    We use this to derive the time evolution operator (i.e. {\em propagator}) for the Dirac IVP, and write the solution as
    \begin{equation} \label{propE}     \Psi_\text{e}(t_\text{e},s_\text{e})=\mathcal{E}(t_\text{e},s_\text{e})\mathring{\Psi}_\text{e}.
    \end{equation}
 The electronic propagator $\mathcal{E}(t_\text{e},s_\text{e})$ acting on smooth rapidly decaying initial data is defined explicitly by equation (\ref{Dirac Solution}) in the appendix. Using the general results on existence and uniqueness of propagators for normally hyperbolic operators, e.g. \cite{BGPbook} (Theorems 3.2.11 and 3.3.1, and Corollary 3.4.3) one obtains the existence and uniqueness of $\mathcal{E}(t_\text{e},s_\text{e})$ acting on distributions.
 \par
Also relevant is a discussion of the Goursat problem for the Dirac equation. The Goursat problem arises when one seeks to solve a Klein-Gordon equation with data specified on two future null rays emanating from some point. So long as this data is consistent with the Dirac equation, it will produce a solution for the Dirac equation as well. Without loss of generality, we take the point to be $(0,0)$, and write the Goursat problem as 
    \begin{equation}
        \begin{cases}
            \Box_\text{e} \Psi+\omega^2 \Psi&=0 \quad \text{for } |s_\text{e}|<t_\text{e}
            \\
            \Psi(b,b)&=F(b)
            \\
            \Psi(c,-c) &=G(c)
        \end{cases}
    \end{equation}
Written explicitly in the Appendix, the solution has us multiply the (smooth, rapidly decaying) functions $F$ and $G$ by certain propagators and integrate over the intersection of the null-rays with the backwards light-cone of $(t_\text{e},s_\text{e})$:
\begin{figure}[htbp]
        \centering
        \includegraphics[width=4.77cm, height=3cm]{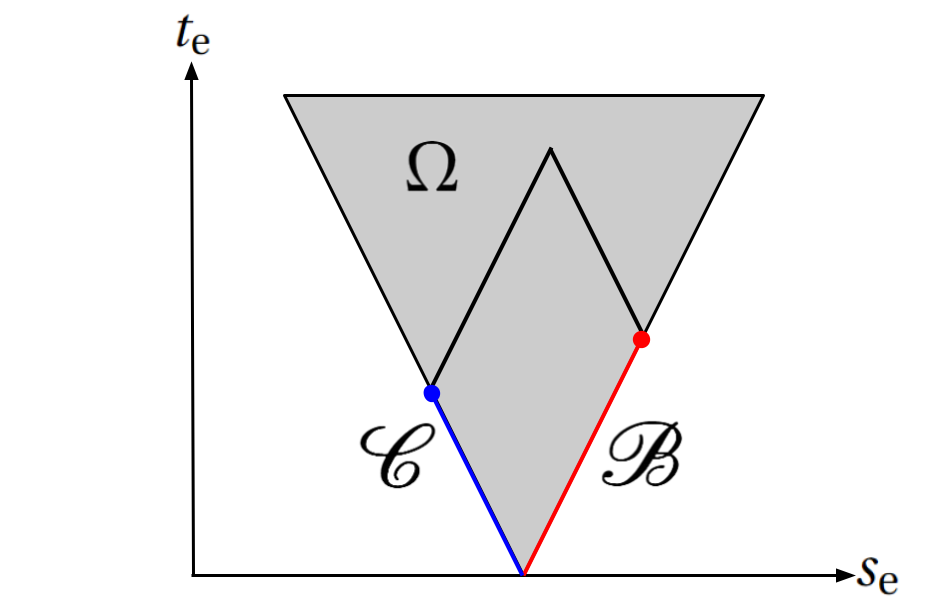}
        \caption{Data along lines $\mathcal{B}$ and $\mathcal{C}$ are sufficient to solve the Goursat problem in $\Omega$}
        \label{fig:Goursat Electron}
    \end{figure}
   \begin{equation}\label{propG}
       \Psi \big{|}_\Omega=G^R(t_\text{e},s_\text{e})F+G^L(t_\text{e},s_\text{e})G,
   \end{equation}
   where $G^R$ and $G^L$ are the right and left time evolution operators of the Goursat problem, their action on smooth functions defined explicitly by equations (\ref{GR Def}), (\ref{GL Def}) in the Appendix, and extended uniquely to distributions in the same manner discussed above for $\mathcal{E}(t_\text{e},s_\text{e})$.
\subsection{The Photon-KTZ Equation}
For a single photon in 1+1 dimensions, the wave function for a photon is shown in \cite{KTZ} to be given by a rank two bi-spinor field defined on the photon's configuration space-time $\mathbb{R}^{1,1}_\text{ph}$, equipped with coordinates $(t_\text{ph},s_\text{ph})$:
\begin{equation}
\Psi_{\text{ph}} = 
\begin{pmatrix}
0 & \chi_- \\
\chi_+ & 0
\end{pmatrix}.
\end{equation}
As a rank-two bi-spinor field, the components of $\Psi_\text{ph}$ transform under the action of $O(1,1)$ as
\begin{equation}
    \chi_{\pm} \xrightarrow{\Lambda} e^{\mp a} \chi_{\pm}, \quad \chi_{\pm} \xrightarrow{\mathbf{P}}  \chi_{\mp}, \quad \chi_{\pm} \xrightarrow{\mathbf{T}}  \chi_{\mp}^*.
\end{equation}
In 1+1 dimensions the KTZ photon wave equation reduces to the massless Dirac equation
\begin{equation}
-i\hbar \gamma^{\mu}\frac{\partial}{\partial x^\mu_\text{ph}}\Psi_\text{ph}=0.
\end{equation}
Given initial data $\mathring{\Psi}_\text{ph}$,
 the solution to the corresponding IVP is given by transports in each component
\begin{equation}\label{Photonic Transport}
    \Psi_\text{ph}(t_\text{ph},s_\text{ph})=\mathcal{P}(t_\text{ph},s_\text{ph})\mathring{\Psi}_\text{ph}:=\begin{pmatrix}
        0 & \mathring{\chi}_-(s_\text{ph}-t_\text{ph})
        \\
        \mathring{\chi}_+(s_\text{ph}+t_\text{ph}) & 0
    \end{pmatrix}.
\end{equation}
$\chi_+$ is transported along \textit{left moving} null rays; while $\chi_-$ is transported along \textit{right moving} null rays.
\subsubsection{Photon probability current}
Similarly to the electron, the photonic wave equation admits a conserved probability current. Given any Killing field $X$ of Minkowski space-time, the manifestly covariant current 
\begin{equation}
    j^\mu_{X}:=\text{tr}(\overline{\Psi}_\text{ph} \gamma^\mu \Psi_\text{ph}\gamma(X))
\end{equation}
is conserved, i.e
\begin{equation}
    \partial_\mu j^\mu_X=0.
\end{equation}
Here $\gamma(X)=\gamma^\mu X_\mu$.
\newline
In \cite{KTZ} it is proven that when $X$ is causal and future-directed, then so is $j_X$, i.e $\eta(j_X,j_X)>0$ and $j^0_X \geq 0$. Thus, this satisfies the appropriate version of the Born rule. However, there is some debate about which time-like Killing field $X$ should be used to define the probability current for a photon. Setting $X=\partial_t$, the time-like unit vector of the "lab frame", returns a probability current formula that looks similar to the well-known conserved current for a spin-half Dirac particle. There are compelling
reasons however for why a quantum probability current for a particle should only depend on the wave function of the particle,
not on some external “observer.” A proposal for an $X$ that only depends on the initial wave function of the system was given in \cite{KTZ}.  For a critique of theories with observer-dependent probability currents, see Struyve
et al. \cite{STRUYVE} and Tumulka \cite{Tumulka_2007}
 (and references therein).
 \par
 Since in this paper we are only interested in solving initial value problems where the wave function is specified on some Cauchy hyperplane given as the zero level-set of a time function $t$ defined on the Minkowski spacetime, we will side-step these arguments by setting up and solving the initial value problem in the fixed frame where $X=\partial_t$. %, which is valid since all of our dynamics are manifestly Lorentz covariant. 
\section{The Three Body El-Ph-El System}
\subsection{The 3-body multi-time wave function}
\begin{comment}

In 1 space dimension the wave function of a photon is a rank-2 bi-spinor field on $\mathbb{R}^{1,1}$ with the following form [KTZ2018]:
\begin{equation}\label{photon wave}
\psi_{\text{ph}} = 
\begin{pmatrix}
0 & \chi_- \\
\chi_+ & 0
\end{pmatrix}
\end{equation}
Where we choose the following as our basis for the subspace of anti-diagonal 2 $\times$ 2 matrices:
\begin{equation}
\mathcal{B}_1 \coloneqq \left\{ E_- = 
\begin{pmatrix}
0 & 1 \\
0 & 0
\end{pmatrix}
,
\hspace{0.2cm}
E_+ =
\begin{pmatrix}
0 & 0 \\
1 & 0
\end{pmatrix}
\right\}
\end{equation}
Meanwhile, the wave function of an electron is a rank-one spinor field on $\mathbb{R}^{1,1}$, according to Paul Dirac [D1928]. In this paper we will have two electrons with the following wave functions:
\begin{equation}\label{e wave}
\psi_{\text{e}_1} =
\begin{pmatrix}
\psi_- \\
\psi_+
\end{pmatrix}
\hspace{1.3cm}
\psi_{\text{e}_2} =
\begin{pmatrix}
\phi_- \\
\phi_+
\end{pmatrix}
\end{equation}
Where we choose the following as our basis for $\mathbb{C}^{2}$:
\begin{equation}
\mathcal{B}_2 \coloneqq \left\{
e_- =
\begin{pmatrix}
1 \\
0
\end{pmatrix}
,
\hspace{0.2cm}
e_+ =
\begin{pmatrix}
0 \\
1
\end{pmatrix}
\right\}
\end{equation}
\end{comment}
Let $\mathcal{M} = \mathbb{R}_{\text{ph}}^{1,1} \times \mathbb{R}_{\text{e}_1}^{1,1} \times \mathbb{R}_{\text{e}_2}^{1,1}$ be the configuration spacetime, where, as will be the case in the rest of the paper, the subscripts $_{\text{ph}}$, $_{\text{e}_1}$, $_{\text{e}_2}$ refer to the photon, the first electron, and the second electron, respectively. We denote by $(x_{\text{ph}}^{\mu}, x_{\text{e}_1}^{\nu},x_{\text{e}_2}^{\kappa})$ where $\mu,\nu,\kappa \in \{0,1\}$ a global system of rectangular coordinates on $\mathcal{M}$. The multi-time wave function of the electron-photon-electron system in one space dimension is a mapping  $\Psi:\mathcal{S}\subset \mathbb{R}^{1,1}_\text{ph}\times \mathbb{R}^{1,1}_{\text{e}_1} \times\mathbb{R}^{1,1}_{\text{e}_2}  \to \mathbb{B}$, where $\mathcal{S}$ is the subset of pairwise spacelike configurations
\begin{equation}
    \mathcal{S}:=\left\{(\textbf{x}_\text{ph}, \textbf{x}_{\text{e}_1},\textbf{x}_{\text{e}_2})\in \mathbb{R}^{1,1}_\text{ph}\times \mathbb{R}^{1,1}_{\text{e}_1}\times \mathbb{R}^{1,1}_{\text{e}_2} : \eta(\textbf{x}_\text{ph}-\textbf{x}_{\text{e}_1},\textbf{x}_\text{ph}-\textbf{x}_{\text{e}_1}) <0, \quad \eta(\textbf{x}_\text{ph}-\textbf{x}_{\text{e}_2},\textbf{x}_\text{ph}-\textbf{x}_{\text{e}_2}) <0 \right\}.
\end{equation}
Here $\eta=\mbox{diag}\,(1,-1)$ is the Minkowski metric on $\mathbb{R}^{1,1}$ and $\mathbb{B}$ is the vector bundle over $\mathcal{S}$ whose fibers are isomorphic to the set of anti-diagonal matrices in $M_2(\mathbb{C})$ tensored with the vector space $\mathbb{C}^2$ twice.
We take a basis for this bundle by taking tensor products of the basis elements for the photon wave function $E_{\varsigma}$ with the basis elements for the Dirac wave function $e_\varsigma$. In this basis,$\Psi$ can be expanded as 
\begin{equation}\label{3-body-psi}
\Psi(\mathbf{x}_{\text{ph}},\mathbf{x}_{\text{e}_1},\mathbf{x}_{\text{e}_2}) = \sum_{\varsigma_{0},\varsigma_{1},\varsigma_{2}\in \{-,+\}} \psi_{\varsigma_{0} \varsigma_{1} \varsigma_{2}}(\mathbf{x}_{\text{ph}},\mathbf{x}_{\text{e}_1},\mathbf{x}_{\text{e}_2}) E_{\varsigma_{0}} \otimes e_{\varsigma_1} \otimes e_{\varsigma_2}
\end{equation}
where the eight components given by $\psi_{\varsigma_0 \varsigma_1 \varsigma_2}$ are complex valued functions on $\mathcal{M}$. These transform as
\begin{equation}
\psi_{\varsigma_0 \varsigma_1 \varsigma_2} \xrightarrow{\Lambda} e^{-a(\varsigma_0 1 + \varsigma_1 \frac{1}{2} +\varsigma_2 \frac{1}{2})}\psi_{\varsigma_0 \varsigma_1 \varsigma_2}, \hspace{0.2cm} \psi_{\varsigma_0 \varsigma_1 \varsigma_2} \xrightarrow{\mathbf{P}} \psi_{\overline{\varsigma}_0 \overline{\varsigma}_1 \overline{\varsigma}_2}, \hspace{0.2cm} \psi_{\varsigma_0 \varsigma_1 \varsigma_2} \xrightarrow{\mathbf{T}} \psi_{\overline{\varsigma}_0 \overline{\varsigma}_1 \overline{\varsigma}_2}^*
\end{equation}
where $\Lambda$ denotes the proper Lorentz transformations, $\mathbf{P}$ space reflections and $\mathbf{T}$ time reflections. $\overline{\varsigma}$ denotes the sign opposite of $\varsigma$, and $^*$ denotes complex conjugation extended in the natural way for tensor space.
\subsection{The 3-body multi-time wave equations}
Let 
\begin{equation}
\gamma_{\text{ph}}^\mu \coloneqq \gamma^\mu \otimes \mathbbm{1} \otimes \mathbbm{1}, \hspace{1cm}
\gamma_{\text{e}_1}^{\nu} \coloneqq \mathbbm{1} \otimes \gamma^\nu \otimes \mathbbm{1} \hspace{1cm}
\gamma_{\text{e}_2}^\kappa \coloneqq \mathbbm{1} \otimes \mathbbm{1} \otimes \gamma^\kappa
\end{equation}
% where $\{\gamma^0,\gamma^1\}\subset M_2(\mathbb{C})$ are Dirac matrices satisfying 
% \begin{equation}
% \gamma^{\mu} \gamma^{\nu} + \gamma^{\nu} \gamma^{\mu} = 2 \eta^{\mu \nu} \mathbb{1},
% \end{equation}
and
\begin{equation} 
D_{\text{ph}} \coloneqq \gamma_{\text{ph}}^\mu \partial_{x_{\text{ph}}^\mu} \hspace{1.3cm}
D_{\text{e}_1} \coloneqq \gamma_{\text{e}_1}^\nu \partial_{x_{\text{e}_1}^\nu} \hspace{1.3cm}
D_{\text{e}_2} \coloneqq \gamma_{\text{e}_2}^\kappa \partial_{x_{\text{e}_2}^\kappa},
\end{equation}
where repeated indices imply Einstein summation.

The multi-time dynamics of the electron-photon-electron wave function are then defined by the following:
\begin{enumerate}
%\begin{enumerate}
\item The free \textit{multi-time equations} on $\mathcal{S}$
\begin{align}\label{multi-time equation}
\left\{\begin{array}{rcl}
-i\hbar D_{\text{ph}} \Psi&=&0 \\
-i\hbar D_{\text{e}_1} \Psi + m_\text{e} \Psi&=&0 \\ 
-i\hbar D_{\text{e}_2} \Psi + m_\text{e} \Psi&=&0 
\end{array}\right..
\end{align}
\item \textit{Initial data} specified on the surface
\begin{equation}
\mathcal{I} \coloneqq \left\{ (t_{\text{ph}},s_{\text{ph}},t_{\text{e}_1},s_{\text{e}_1},t_{\text{e}_2},s_{\text{e}_2})\in \overline{\mathcal{S}} : t_{\text{ph}}=t_{\text{e}_1}=t_{\text{e}_2}=0\right\},
\end{equation}
namely
\begin{equation}\label{initial data}
\Psi (0,s_{\text{ph}},0,s_{\text{e}_1},0,s_{\text{e}_2}) = \mathring{\Psi} (s_{\text{ph}},s_{\text{e}_1},s_{\text{e}_2}).
\end{equation}
\item \textit{Boundary conditions} on the coincidence sets $\mathcal{C}_1$ and $\mathcal{C}_2$. These implement contact interactions, and their exact form will be determined in the next section by analyzing the conditions for probability conservation. Since we are interested in studying the dynamics of a photon between two electrons, we will restrict our attention to two disjoint subsets of $\mathcal{S}$
\begin{align}\label{S1 definition}
&\mathcal{S}_1\coloneqq
\left\{(t_{\text{ph}},s_{\text{ph}},t_{\text{e}_1},s_{\text{e}_1},t_{\text{e}_2},s_{\text{e}_2})\in \mathcal{S}:s_{\text{e}_1}<s_{\text{ph}}<s_{\text{e}_2}\right\}\\ &\mathcal{S}_2\coloneqq
\left\{(t_{\text{ph}},s_{\text{ph}},t_{\text{e}_1},s_{\text{e}_1},t_{\text{e}_2},s_{\text{e}_2})\in \mathcal{S}:s_{\text{e}_2}<s_{\text{ph}}<s_{\text{e}_1}\right\},
\end{align}
and the boundary conditions will be linear relations between limits along events in $\mathcal{S}_1$ and $\mathcal{S}_2$ towards the sets $\mathcal{C}_1$ and $\mathcal{C}_2$ where
\begin{equation}\label{C1 and C2 definition}
\begin{split}
\mathcal{C}_1 \coloneqq 
\left\{(t_{\text{ph}},s_{\text{ph}},t_{\text{e}_1},s_{\text{e}_1},t_{\text{e}_2},s_{\text{e}_2})\in \partial\mathcal{S}:s_{\text{e}_1}=s_{\text{ph}}, \hspace{0.2cm} t_{\text{e}_1}=t_{\text{ph}} \right\}\\
\mathcal{C}_2 \coloneqq 
\left\{(t_{\text{ph}},s_{\text{ph}},t_{\text{e}_1},s_{\text{e}_1},t_{\text{e}_2},s_{\text{e}_2})\in \partial\mathcal{S}:s_{\text{e}_2}=s_{\text{ph}}, \hspace{0.2cm} t_{\text{e}_2}=t_{\text{ph}}\right\}
\end{split}
\end{equation}
\end{enumerate}
%\end{enumerate}
\subsection{Conserved 3-body currents and probability-preserving boundary conditions}
For a pure product state $\Psi=\psi_\text{ph}\otimes \psi_{\text{e}_1}\otimes \psi_{\text{e}_2}$, we define its \textit{Dirac adjoint} $\overline{\Psi}$ as
\begin{equation}
\overline{\Psi} \coloneqq \overline{\psi_{\text{ph}}} \otimes \overline{\psi_{\text{e}_1}} \otimes \overline{\psi_{\text{e}_2}} = \gamma^{0}\psi_{\text{ph}}^{\dagger}\gamma^{0} \otimes \psi_{\text{e}_1}^{\dagger}\gamma^{0} \otimes \psi_{\text{e}_2}^{\dagger}\gamma^{0}.
\end{equation}
For a general three-time, three-body wave function we extend this definition by conjugate linearity. 
\begin{comment}
Note that since $\Psi$ satisfies (\ref{multi-time equation}), $\overline{\Psi}$ similarly satisfies:
\begin{align}\label{conjugate multi-time eq}
\begin{cases}
i\hbar \partial_{x_{\text{ph}}^{\mu}} \overline{\Psi}\gamma_{\text{ph}}^{\mu}&=0 \\
i\hbar \partial_{x_{\text{e}_1}^{\nu}} \overline{\Psi}\gamma_{\text{e}_1}^{\nu} + m_{e} \overline{\Psi}&=0 \\ 
i\hbar \partial_{x_{\text{e}_2}^{\kappa}} \overline{\Psi}\gamma_{\text{e}_2}^{\kappa} + m_{e} \overline{\Psi}&=0
\end{cases}
\end{align}
\end{comment}
We take the joint probability current of our three-time wave function to be the tensor product of the three particle's individual probability currents. Let $X$ be a time-like Killing field on $\mathbb{R}^{1,1}$, we define the joint probability tensor $j^{\mu \nu \kappa}_X$ as
\begin{equation}\label{3-current}
j^{\mu\nu\kappa}_{X}[\Psi] \coloneqq \frac{1}{4} \text{tr}_{\text{ph}} \left\{ \overline{\Psi} \gamma_{\text{ph}}^{\mu} \gamma_{\text{e}_1}^{\nu} \gamma_{\text{e}_2}^{\kappa} \Psi \gamma_{\text{ph}}(X)\right\},
\end{equation}
where $\text{tr}_{ph} = \text{tr} \otimes \mathbbm{1} \otimes \mathbbm{1}$ is the operation of taking the trace of the photonic component. By a straightforward generalization of the two-particle results in \cite{KLTZ}, namely Propositions 4.2, 4.3, 4.5 and Theorem 4.6, to three-particle systems, one obtains the following:

\begin{comment}
We may write the components of $j_X$ in terms of the components of $\Psi$ as
\begin{equation}\label{j expanded}
\begin{split}
j^{\mu\nu\kappa}_{X}[\Psi] = \frac{1}{4}\sum_{\rho = 0}^{1}X_{\rho}[&|\psi_{---}|^{2} + (-1)^k|\psi_{--+}|^{2} + (-1)^{\nu}|\psi_{-+-}|^{2} + (-1)^{\nu + \kappa}|\psi_{-++}|^{2}
\\
+(-1)^{\mu+\rho}|\psi_{+--}&|^{2} + (-1)^{\mu+\kappa+\rho}|\psi_{+-+}|^{2} + (-1)^{\mu+\nu+\rho}|\psi_{++-}|^{2} + (-1)^{\mu+\nu+\kappa+\rho}|\psi_{+++}|^{2}].
\end{split}
\end{equation}
\end{comment}

\begin{enumerate}
    \item 
Let $\Psi$ be a $C^{1}$ solution of (\ref{multi-time equation}) and let $X$ be any constant vector field on $\mathbb{R}^{1,1}$. Then the current $j_{X}$ is jointly conserved, that is:
\begin{equation}\label{j joint-conservation}
\begin{split}
\begin{cases}
\partial_{x_{\text{ph}}^{\mu}} j^{\mu\nu\kappa}_{X} &= 0
\hspace{0.2cm}
\text{for}
\hspace{0.2cm}
\nu,\kappa \in \{0,1\}
\\
\partial_{x_{\text{e}_1}^{\nu}} j^{\mu\nu\kappa}_{X} &= 0
\hspace{0.2cm}
\text{for}
\hspace{0.2cm}
\mu,\kappa \in \{0,1\}
\\
\partial_{x_{\text{e}_2}^{\kappa}} j^{\mu\nu\kappa}_{X} &= 0
\hspace{0.2cm}
\text{for}
\hspace{0.2cm}
\mu,\nu \in \{0,1\}.
\end{cases}
\end{split}
\end{equation}
\item
Let the vector field $X$ be time-like. Then the only Poincaré invariant class of boundary conditions for the components of $\Psi$ that conserves the total probability
\begin{equation}\label{probability conservation}
\int_{\Sigma^3 \cap \mathcal{S}_i} \omega_j=
\int_{\Sigma^3 \cap \mathcal{S}_i} \sum_{\mu,\nu,\kappa \in \{0,1\}} (-1)^{\mu+\nu+\kappa} j_{X}^{\mu \nu \kappa}(\textbf{x}_{\text{ph}},\textbf{x}_{\text{e}_1},\textbf{x}_{\text{e}_2}) dx_{\text{ph}}^{1-\mu} \wedge dx_{\text{e}_1}^{1-\nu} \wedge dx_{\text{e}_2}^{1-\kappa}
\end{equation}
independently of Cauchy surface $\Sigma$ is: 
\begin{equation}\label{psi boundary}
\begin{split}
&\lim_{\varepsilon \to 0}\psi_{-+\varsigma_2}(t,s\pm\varepsilon,t,s\mp\varepsilon,\textbf{x}_{\text{e}_2})=
\lim_{\varepsilon \to 0}e^{\pm i\theta_{1}}\sqrt{\frac{X^0+X^1}{X^0-X^1}}\psi_{+-\varsigma_2}(t,s\pm\varepsilon,t,s\mp\varepsilon,\textbf{x}_{\text{e}_2})
\\
&\lim_{\varepsilon \to 0}\psi_{+\varsigma_1-}(t,s\mp\varepsilon,\textbf{x}_{\text{e}_1},t,s\pm\varepsilon)=
\lim_{\varepsilon \to 0}e^{\pm i\theta_{2}}\sqrt{\frac{X^0-X^1}{X^0+X^1}}\psi_{+\varsigma_1-}(t,s\mp\varepsilon,\textbf{x}_{\text{e}_1},t,s\pm\varepsilon),
\end{split}
\end{equation}
for constant phases $\theta_{1},\theta_{2}  \in [0,2\pi)$. 
The additional anti-symmetry condition imposed by the Pauli exclusion, i.e.
\begin{equation}
\psi_{\varsigma_0\varsigma_1\varsigma_2}(\textbf{x}_\text{ph},\textbf{x}_{\text{e}_1},\textbf{x}_{\text{e}_2})=-\psi_{\varsigma_0\varsigma_2\varsigma_1}(\textbf{x}_\text{ph},\textbf{x}_{\text{e}_2},\textbf{x}_{\text{e}_1}) 
\end{equation}
implies that $\theta_1=\theta_2$.
\end{enumerate}
%\end{theorem}
We now have all the tools needed to set up an initial-boundary value problem for the dynamics of our three body wave function. But before we do so, let us review the results of \cite{KLTZ} concerning the dynamics of a two-body electron-photon wave function.
\section{Review of IBVP for two-body wave functions}
In this section we review some of the results of \cite{LiNi2015}, \cite{LiNi2020}, and \cite{KLTZ} for the initial-boundary value problem (IBVP) for a two-time, two-body wave function that are relevant to our case. The IBVP studied by these authors consists of the free two-time equations on a subset of the two-particle configuration space-time \begin{align}\label{2body S1 definitionL}
&\mathcal{S}_1\coloneqq\left\{(t_{\text{ph}},s_{\text{ph}},t_{\text{e}},s_{\text{e}})\in \mathbb{R}^{(1,1)}_\text{ph}\times \mathbb{R}^{(1,1)}_\text{e}:\eta(\textbf{x}_\text{ph}-\textbf{x}_\text{e},\textbf{x}_\text{ph}-\textbf{x}_\text{e}) <0, \quad s_{\text{ph}}<s_{\text{e}}\right\}
\end{align} combined with a probability conserving boundary condition.  For simplicity, we only state the results in the fixed Lorentz frame where the timelike vectorfield $X$ in \cite{KLTZ} is $\partial_t$. These authors have shown:
\begin{enumerate}
    \item 
(\cite[Lemma 5.1]{LiNi2015}, \cite[Thm. 4.6]{KLTZ})  %Let the vector field $X$ be time-like. 
The only Poincaré invariant class of boundary conditions for the components of $\Psi$ that preserves the total probability in $\mathcal{S}_1$ is: 
\begin{equation}\label{2 body psi boundaryL}
\lim_{\varepsilon \to 0}\psi_{+-}(t,s\mp\varepsilon,t,s\pm\varepsilon)=
\lim_{\varepsilon \to 0}e^{\pm i\theta}
%\sqrt{\frac{X^0-X^1}{X^0+X^1}}
\psi_{-+}(t,s\mp\varepsilon,t,s\pm\varepsilon).
\end{equation}
for a constant  $\theta  \in [0,2\pi)$. 

\item \label{2Part Well} (\cite[Thm. 5.1]{KLTZ}, \cite[Thm. 4.3]{LiNi2020})
Let $\theta \in [0,2\pi)$. Let $\mathcal{S}_1$ denote the set of space-like configurations in $\mathcal{M}$, as was defined in (\ref{2body S1 definitionL}), let $\mathcal{C} \subset \partial \mathcal{S}_1$ be the coincidence set defined by $\{ \textbf{x}_\text{ph}=\textbf{x}_\text{e}\}$, and let $\mathcal{I}$ be the initial surface. Let  the function $\mathring \Psi : \mathcal{I} \rightarrow \mathbb{C}^4$ be $C^1$ and compactly supported in the half-space $\mathcal{I} \cap \overline{\mathcal{S}_1}$, and suppose that the initial-value function also satisfies the boundary conditions given by (\ref{2 body psi boundaryL}). Then the following initial-boundary value problem for the wave function $\Psi: \mathcal{M} \rightarrow \mathbb{C}^4$,
\begin{equation}\label{two body problem}
\left\{
\begin{array}{rclr} 
-i\hbar D_{\text{ph}} \Psi &= &0  & 
\\
-i\hbar D_{\text{e}} \Psi + m_\text{e} \Psi &= &0  &\text{in } \mathcal{S}_1 
\\
\Psi &= &\mathring \Psi  &\text{on } \mathcal{I} 
\\
\psi_{+-} &= &e^{i\theta}\psi_{-+} \hspace{0.5cm}&\text{on } \mathcal{C}
\end{array}
\right.
\end{equation}
has a unique global-in-time solution that is supported in $\overline{\mathcal{S}_1}$, and depends continuously on the initial data $\mathring \Psi$.
\end{enumerate}
\begin{remark}
    For an outline of the fixed-point argument for existence given by \cite{LiNi2020} we direct the reader to Theorem~\ref{thm:2body} in the Appendix. The fixed-point argument is simpler than the proof given in \cite{KLTZ}, but it is not constructive. The explicit solution constructed in \cite{KLTZ}, as we shall see later, can be elegantly represented using Feynman-like diagrams. On the other hand the proof given in \cite{LiNi2020} applies to a much larger class of initial-boundary value problems, which will be useful when we generalize to the 3-particle case of a photon between two electrons.
\end{remark}
\subsection{Construction of the explicit 2-body solution}
    Here we outline the proof of Theorem 5.1 in \cite{KLTZ}, where they establish that $\Psi$ has a unique global-in-time solution.  They do this by constructing an explicit solution of the IBVP for suitably regular initial data, e.g. data belonging to the Sobolev space $\dot{H}^s(\mathbb{R}^2)$ for $s\geq 2$, which allows for standard energy estimates to be used to prove uniqueness.  The solution can be expressed in terms of the electronic and photonic propagators introduced in \eqref{propE}, \eqref{propG}, and \eqref{Photonic Transport}. Extending the solution operator to distributional data and proving its uniqueness is again possible thanks to the previously cited results in \cite{BGPbook}.
    
    We begin by noting that by finite speed of propagation any data evolving sufficiently far from the boundary $\mathcal{C}$ is not influenced by the boundary condition, and is only governed by the free multi-time equations. So, we begin by decomposing $\mathcal{S}_1$ into two disjoint parts
\begin{equation}
    \mathcal{S}_1 = \mathcal{F}\cup \mathcal{N}
\end{equation}
via
\begin{equation}
    \mathcal{F}:= \left\{(\textbf{x}_\text{ph},\textbf{x}_\text{e}) \in \mathcal{S}_1 : s_\text{ph}+t_\text{ph}\leq s_\text{e}-t_\text{e} \right\},
\end{equation}
\begin{equation}
     \mathcal{N}:= 
     \left\{(\textbf{x}_\text{ph},\textbf{x}_\text{e}) \in \mathcal{S}_1 : s_\text{ph}+t_\text{ph}> s_\text{e}-t_\text{e} \right\},
\end{equation}
and give distinct solution formulas in each region. Data evolving on $\mathcal{F}$ is governed only by the free equations, and its solution formula will be identical to the solution formula for a freely evolving two-body system. This formula can easily be written in terms of the time evolution operators of photonic transport and  electronic Dirac equation. 
Let $\mathcal{P}(t_\text{ph},s_{\text{ph}})$ denote the photonic evolution operator defined by equation (\ref{Photonic Transport}),
and $\mathcal{E}(t_\text{e},s_\text{e})$ denote the electronic evolution operator defined by equation (\ref{Dirac Solution}) in the Appendix.  
%\begin{proposition}    
In \cite[Eqs. (5.11--5.14)]{KLTZ},  the authors obtain explicit formulas for the unique, freely evolving two-time, two-body wavefunction of the system in terms of one-body photonic transport $\mathcal{P}$  and electronic propagator $\mathcal{E}$, when the initial data is suitably regular. Using standard  density arguments combined with the existence and uniqueness results in \cite{BGPbook} for the fundamental solution of normally hyperbolic operators that were cited in the above, one sees that the free 2-body propagator acting on distributions is the tensor product of the two single body time evolution operators, namely, 
%\end{proposition}
for $(\textbf{x}_\text{ph},\textbf{x}_\text{e})\in \mathcal{F}$ the solution formula for $\Psi$ is given by
\begin{equation}\label{Far SolL}
    \Psi(\textbf{x}_\text{ph},\textbf{x}_\text{e})=\big( \mathcal{P}( \textbf{x}_\text{ph}) \otimes \mathcal{E}(\textbf{x}_\text{e})\big) \mathring{\Psi}.
\end{equation}
Our goal now is to find a solution formula in $\mathcal{N}$. Fix $(\textbf{x}_\text{ph},\textbf{x}_\text{e})\in \mathcal{N}$. Looking closely at the explicit free evolution solution formula given by equation ($\ref{Far SolL}$),  we find that this formula is still valid in the near configuration regime for the $\varsigma_0=-$ components, i.e for $\psi_{--}$ and $\psi_{-+}$. This is because these components transport the photon along right-moving null-rays, so in $\mathcal{S}_1$ this evolution cannot depend on any data along $\mathcal{C}$. Letting $\Pi_{\varsigma_0\varsigma_1 }$ be the projection operator sending $\Psi$ to $\psi_{\varsigma_0 \varsigma_1}$, the solution formula for $\psi_{-\varsigma_1}$ on $\mathcal{N}$ is written as
\begin{equation}\label{Near free SolL}
    \psi_{-\varsigma_1}(\textbf{x}_\text{ph},\textbf{x}_\text{e})=\Pi_{-\varsigma_1}\big( \mathcal{P}( \textbf{x}_\text{ph}) \otimes \mathcal{E}( \textbf{x}_\text{e})\big) \mathring{\Psi}.
\end{equation}
To solve for the remaining components, consider the forwards light-cone emanating from the point $(0,s_\text{ph}+t_\text{ph})\in \mathbb{R}^{1,1}$. For all $\textbf{x}_\text{e}'$ which are on the right edge of this cone, we have data for  $\psi_{+\varsigma_1}(\textbf{x}_\text{ph},\textbf{x}_\text{e}')$ using equation (\ref{Far SolL}). Since $\psi_{+\varsigma_1}$ each satisfy an electronic Klein-Gordon equation, we are motivated to write down and solve a Goursat problem in the electronic variables. For this we will need data along two lines in configuration space-time
\begin{equation}
    \mathcal{B}=\left\{ (\textbf{x}_\text{ph},\textbf{x}_\text{e}') \in \mathcal{S}_1: \textbf{x}_\text{e}'=(b,s_\text{ph}+t_\text{ph}+b), \quad b \in \mathbb{R}\right\}
\end{equation}
\begin{equation}
    \mathcal{C}'=\left\{ (\textbf{x}_\text{ph},\textbf{x}_\text{e}') \in \mathcal{S}_1: \textbf{x}_\text{e}'=(c,s_\text{ph}+t_\text{ph}-c), \quad c \in \mathbb{R}\right\}
\end{equation}
\begin{figure}[htbp]
        \centering
        \includegraphics[width=4.77cm, height=3cm]{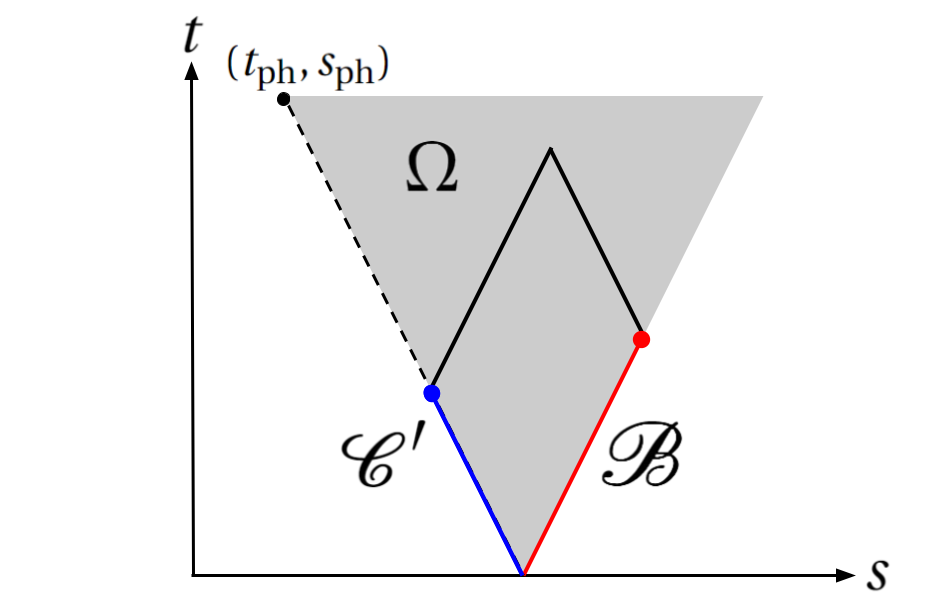}
        \caption{Keeping $\textbf{x}_\text{ph}$ fixed, data for $\textbf{x}_\text{e}$ along the lines $\mathcal{B}$ and $\mathcal{C}$ is sufficient to solve the Goursat problem for all $\textbf{x}_\text{e} \in \Omega$.}
        \label{fig:El-Ph GoursatL}
    \end{figure}
\par
To solve for the data along $\mathcal{C}'$, we appeal to the fact that $\varsigma_0=+$ components are constant as we transport the photonic configuration along left-moving null rays. Letting $(\textbf{x}_\text{ph},\textbf{x}_\text{e}')$ be an arbitrary point in $\mathcal{C}'$, the photonic transport equations imply
\begin{equation}\label{C setupL}
    \psi_{+-}(\textbf{x}_\text{ph},\textbf{x}_\text{e}')=\psi_{+-}(\textbf{x}_\text{e}',\textbf{x}_\text{e}')=e^{i\theta}\psi_{-+}(\textbf{x}_\text{e}',\textbf{x}_\text{e}').
\end{equation}
The second equality of equation (\ref{C setupL}) will be our only application of the contact boundary condition. Abusing notation slightly, we may say that $\psi_{+-}\big{|}_\mathcal{C'}=e^{i\theta}\psi_{-+}\big{|}_\mathcal{C}$. 
\par
Applying the Goursat problem solution formula (\ref{goursat solution}) derived in the appendix,  the solution for $\psi_{+-}$ on near configurations is written as 
\begin{equation}\label{E-Ph GoursatL}
    \psi_{+-}(\textbf{x}_\text{ph},\textbf{x}_\text{e})= G^{R}(t_\text{e},\tilde{s}) \big(\psi_{+-} \big{|}_\mathcal{B} \big)+ G^{L}(t_\text{e},\tilde{s}) \big(e^{i\theta} \psi_{-+} \big{|}_\mathcal{C} \big)
\end{equation}
where $\tilde{s}=s_\text{e}-(s_\text{ph}+t_\text{ph})$. To write this solution in terms of the initial data, we may plug equations (\ref{Far SolL}) and (\ref{Near free SolL}) into 
(\ref{E-Ph GoursatL})
for $\psi_{+-} \big{|}_\mathcal{B}$ and $e^{i\theta} \psi_{-+} \big{|}_\mathcal{C}$ respectively. 
\par
To solve for the fourth and final component $\psi_{++}$, we plug equation (\ref{E-Ph GoursatL}) into the sourced transport equation that this component satisfies and receive
\begin{equation}\label{Ph-El Sourced TransportL}
    \psi_{++}(\textbf{x}_\text{ph},\textbf{x}_\text{e})= \psi_{++}(\textbf{x}_\text{ph},0,s_\text{e}+t_\text{e})-i\omega \int_{0}^{t_\text{e}}\psi_{+-}(\textbf{x}_\text{ph},\tau, s_\text{e}+t_\text{e}-\tau) d\tau
\end{equation}
This completes our algorithm to solve for the two-body wave function everywhere in $\mathcal{S}_1$. To save space, we will refer to this general process as a "Contact Evolution Operator" $C_\theta(\textbf{x}_\text{ph},\textbf{x}_\text{e})$.
\subsection{Feynman-like Diagrams for Photon-Electron Pairs}
The evolution of our interacting two-body system can be more elegantly represented using Feynman diagrams. We found that the evolution of the two-body wave function for "Far" configurations is identical to the free evolution.
\begin{figure}[htbp]
        \centering
        $\Psi\big{|}_\mathcal{F}
        $\qquad $\mathlarger{=}$ \qquad \includegraphics[width=1.8cm, height=1cm, valign=c]{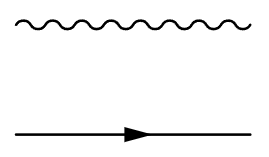} $\qquad$ $\mathlarger{=}$ $\qquad$ $(\mathcal{P}\otimes \mathcal{E}) \mathring{\Psi} $
        \caption{We sum over all diagrams which depicts the electron and photon freely propagating.}
    \end{figure}
\par
On configurations $(\textbf{x}_\text{ph},\textbf{x}_\text{e})$ which are "Near", two of the wave function components no longer satisfy the free evolution. The $\varsigma_0=+$ components become a sum of two terms. The first is non-interacting data: $\varsigma_0=+$ data on $\mathcal{B}$ which has freely propagated to $(\textbf{x}_\text{ph},\textbf{x}_\text{e})$. The other term arose purely from the interaction: $\varsigma_0=-$ data which has propagated through $\mathcal{C}$, but was reborn in the form of $\varsigma_0=+$ and then freely propagated to $(\textbf{x}_\text{ph},\textbf{x}_\text{e})$. \begin{figure}[htbp]
\centering
$\Psi\big{|}_{\mathcal{N}}$ \qquad $=$ \qquad
\includegraphics[width=1.8cm, height=1cm, valign=c]{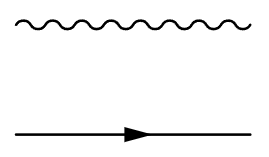}
        \quad $\mathlarger{+} \quad e^{i\theta}$  \includegraphics[width=1.8cm, height=1cm, valign=c]{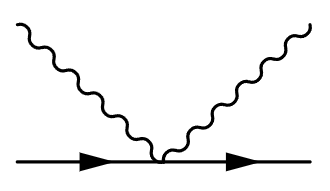}$\qquad$ $\mathlarger{=}$  $\begin{pmatrix}
            \Pi_{--} (\mathcal{P}\otimes \mathcal{E}) \mathring{\Psi}
            \\
            \Pi_{-+} (\mathcal{P}\otimes \mathcal{E}) \mathring{\Psi}
            \\
            G^{R} \big(\psi_{+-} \big{|}_\mathcal{B} \big)
            \\
            \psi_{++}\big{|}_{\mathcal{B}} -i\omega \int_0^{t_e} G^{R} \big(\psi_{+-} \big{|}_\mathcal{B} \big)
        \end{pmatrix} +e^{i \theta} \begin{pmatrix}
            0
            \\
            0
            \\
            G^{L} \big(\psi_{-+} \big{|}_\mathcal{C} \big)
            \\
            -i\omega \int_0^{t_e} G^{L} \big(\psi_{-+} \big{|}_\mathcal{C} \big)
        \end{pmatrix}$
        \caption{There are now two diagrams we must sum over: freely propagated data, and data resulting from a collision.}
\end{figure}
\begin{remark}
Restricting the two-body wave function to $\mathcal{S}_1$ has produced a contact interaction between the photon and the electron that is worthy of being called "Compton Scattering in $1+1$ dimensions". The electron and the photon propagate freely at first, until the photon inevitably reaches contact with the electron at which point the photon is forced to bounce off the electron and propagate away in the opposite direction. The number of photons remains fixed during this interaction, with the photon before the incident being the exact same as the one after, as originally desired by Darwin. 
\end{remark}
In the following sections, we will investigate how these contact boundary interactions generalize to the case of one photon between two electrons. To simplify our proofs we will write our solutions in terms of Feynman-like diagrams. A diagram in which an electron-photon pair become incident and then bounce off means that we are applying a contact evolution operator $C_\theta$ in those variables. This means we are setting up and solving a Goursat problem for the appropriate wave function component, with data on one null line representing freely propagating data while data on the other is given by the boundary condition, and thus represents interacting data. We then plug the Goursat solution into a sourced transport to solve for the other component.
\section{IBVP for three-time three-body wave function}
In this section we consider the initial-boundary value problem (IBVP) for the three-time three-body wave function consisting of equations (\ref{multi-time equation}), (\ref{psi boundary}). For simplicity, we work in the fixed Lorentz frame where $X=\partial_t$.
\begin{theorem}
Let $\theta_1,\theta_2 \in [0,2\pi)$ be two constant phases. Let $\mathcal{S}_1$ denote the set of space-like configurations in $\mathcal{M}$, as was defined in (\ref{S1 definition}), $\mathcal{C}_1 \subset \partial \mathcal{S}_1$ and $\mathcal{C}_2 \subset \partial \mathcal{S}_1$ be the coincidence sets defined by (\ref{C1 and C2 definition}), and $\mathcal{I}$ be the initial surface as was defined in (\ref{initial data}). Let  the function $\mathring \Psi : \mathcal{I} \rightarrow \mathbb{C}^8$ be $C^1$ and compactly supported in the half-space $\mathcal{I} \cap \overline{\mathcal{S}_1}$, and suppose that the initial wave function also satisfies the boundary conditions given by (\ref{psi boundary}). Then the following initial-boundary value problem for the wave function $\Psi: \mathcal{M} \rightarrow \mathbb{C}^8$,
\begin{equation}\label{three body problem}
\left\{
\begin{array}{rclr} 
-i\hbar D_{\text{ph}} \Psi &= &0  & 
\\
-i\hbar D_{\text{e}_1} \Psi + m_{e} \Psi &= &0 & 
\\ 
-i\hbar D_{\text{e}_2} \Psi + m_\text{e} \Psi &= &0  &\text{in } \mathcal{S}_1 
\\
\Psi &= &\mathring \Psi  &\text{on } \mathcal{I} 
\\
\psi_{-+\varsigma_2} &= &e^{i\theta_1}\psi_{+-\varsigma_2} \hspace{0.5cm}&\text{on } \mathcal{C}_1
\\
\psi_{+\varsigma_1-} &= &e^{i\theta_2}\psi_{-\varsigma_1+}  &\text{on } \mathcal{C}_2
\end{array}
\right.
\end{equation}
has a unique global-in-time solution that is supported in $\overline{\mathcal{S}_1}$, is continuous with bounded derivatives, and depends continuously on the initial data $\mathring \Psi$.
\end{theorem}
\begin{proof}
    For a fixed-point argument of existence and uniqueness modeled on the one in \cite{LiNi2020} we direct the reader to Appendix \ref{AppendixLN}. 
\end{proof}
We would like to provide an algorithm for computing the multi-time wave function which admits a Feynman diagram interpretation in the same sense that we had for the photon-electron system. 
\subsection{The solution in terms of Feynman diagrams  }
 Our attempt to construct an algorithm for $\Psi$ will repeat the exact same steps as the ones taken for two-body interacting photon-electron wave function.
%Importantly, our picture of the boundary condition and how it produces interaction will carry over nicely to the three-body case. We will keep in mind that the $\varsigma_0=-$ components seek to transport the photon to the right, meaning this data will eventually leave $\mathcal{S}_1$ through the $\mathcal{C}_2$ boundary. As this data crosses through $\mathcal{C}_2$, it will be reborn in the form of $\varsigma_0=+$ data and freely propagate away. Similarly, $\varsigma_0=+$ data which propagates through $\mathcal{C}_1$ is reborn in the form of $\varsigma_0=-$ data, and then propagates away.
We begin by splitting up our configuration space-time into relevant pieces, which are analogues to the "Far" and "Near" configurations from the two-body construction.
    \newline
    \textbf{Free Configurations:} Configurations in $\mathcal{S}_1$ for which none of the backwards light-cones of the particles cross. 
    \begin{equation}
    \mathcal{F}:= \left\{(\textbf{x}_\text{ph},\textbf{x}_{\text{e}_1},\textbf{x}_{\text{e}_2}) \in \mathcal{S}_1 : s_{\text{e}_1}+t_{\text{e}_1}\leq s_\text{ph}-t_\text{ph}, \quad  s_\text{ph}+t_\text{ph}\leq s_{\text{e}_2}-t_{\text{e}_2} \right\}
\end{equation}
\newline
\textbf{Compton Configurations:} Configurations for which at least one of the backwards light cones of the electrons cross the photon's, but the electron's cones do not cross each other.
\begin{equation}
    \mathcal{X} :=\left\{(\textbf{x}_\text{ph},\textbf{x}_{\text{e}_1},\textbf{x}_{\text{e}_2}) \in \mathcal{S}_1 : s_\text{ph}-t_\text{ph}<s_{\text{e}_1}+t_{\text{e}_1}\leq s_{\text{e}_2}-t_{\text{e}_2}\quad \text{or} \quad  s_{\text{e}_1}+t_{\text{e}_1}\leq s_{\text{e}_2}-t_{\text{e}_2} <s_\text{ph}+t_\text{ph}\right\}
\end{equation}
\newline
\textbf{Coulomb Configurations:}
Configurations for which the backwards light-cones of the electrons cross.
\begin{equation}
    \mathcal{S}_1-\big( \mathcal{F} \cup \mathcal{X}\big)= 
    \left\{(\textbf{x}_\text{ph},\textbf{x}_{\text{e}_1},\textbf{x}_{\text{e}_2}) \in \mathcal{S}_1 : s_{\text{e}_1}+t_{\text{e}_1}< s_{\text{e}_2}-t_{\text{e}_2} \right\}
\end{equation}
\subsubsection{Free Configurations}
    Similarly to the two-body case, the solution formula on "Free" configurations is given by the  free three-time evolution operator \begin{equation}
        \Psi(\textbf{x}_\text{ph},\textbf{x}_{\text{e}_1},\textbf{x}_{\text{e}_2})= \big(\mathcal{P}(\textbf{x}_\text{ph})\otimes \mathcal{E}(\textbf{x}_{\text{e}_1}) \otimes  \mathcal{E}(\textbf{x}_{\text{e}_2}) \big) \mathring{\Psi}.
    \end{equation}
    The corresponding diagram is Fig.~\ref{fig:free}.
    \begin{figure}[htbp]
        \centering
        $\Psi\big{|}_\mathcal{F}
        $\qquad $\mathlarger{=}$ \qquad \includegraphics[width=1.8cm, height=1cm, valign=c]{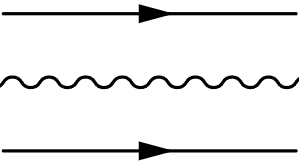}
    \caption{\label{fig:free} Diagram corresponding to free configurations}
    \end{figure}
\subsubsection{Compton Configurations}
For "Compton" configurations, there are three relevant cases for the configurations of our three particles. 
\newline
\textbf{Case 1:} The left electron is "Near" the photon, while the right electron is not. The associated diagram for such configurations is given by Fig.~\ref{fig:compton}
\begin{figure}[htbp]
\centering
\includegraphics[width=1.8cm, height=1cm, valign=c]{Free_E-Ph-E_FMP.png}
        \quad $\mathlarger{+} \quad e^{i\theta_1} $  \includegraphics[width=1.8cm, height=1cm, valign=c]{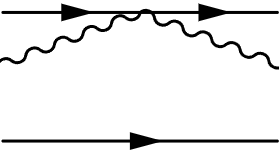}
\caption{\label{fig:compton} Diagram corresponding to Compton configurations}
\end{figure}
\par
From this diagram we determine that the multi-time evolution is given by a contact evolution operator in the photon and electron $1$ variables tensored with the free evolution in the electron $2$ variables
\begin{equation}
    \Psi(\textbf{x}_\text{ph},\textbf{x}_{\text{e}_1},\textbf{x}_{\text{e}_2})= \big(C_{\theta_1}(\textbf{x}_\text{ph},\textbf{x}_{\text{e}_1}) \otimes  \mathcal{E}(\textbf{x}_{\text{e}_2}) \big) \mathring{\Psi}.
\end{equation}
We note that the $\varsigma_0=+$ components of the wave function are freely propagating in this case, and that the $\varsigma_0=-$ components are solved using Goursat evolution operators and sourced transports in electron $1$ variables.
\newline
\textbf{Case 2:} The right electron is "Near" the photon while the left electron is not. In this case electron $1$ freely propagates while the photon undergoes a contact interaction with electron $2$
\begin{equation}
    \Psi(\textbf{x}_\text{ph},\textbf{x}_{\text{e}_1},\textbf{x}_{\text{e}_2})= \big(\mathcal{E}(\textbf{x}_{\text{e}_1}) \otimes C_{\theta_2}(\textbf{x}_\text{ph},\textbf{x}_{\text{e}_2})  \big) \mathring{\Psi}.
\end{equation}
Similarly, we note that the $\varsigma_0=-$ components of the wave function are freely propagating in this case, and that the $\varsigma_0=+$ components is solved by setting up and solving a Goursat problem in the electron $2$ variables.
\newline
\textbf{Case 3:} Configurations in which both backwards light-cones of the electrons cross the backwards light-cone of the photon, but not each other's. On such configurations our wave function is given by a sum of three diagrams, as in Fig.~\ref{fig:contact}: 
\begin{figure}[htbp]
\centering
\includegraphics[width=1.8cm, height=1cm, valign=c]{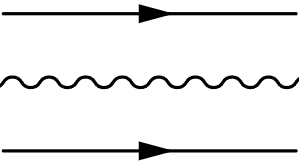}
        \quad $\mathlarger{+} \quad e^{i\theta_1} $  \includegraphics[width=1.8cm, height=1cm, valign=c]{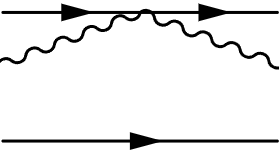}
        \quad $\mathlarger{+} \quad e^{i\theta_2} $
        \includegraphics[width=1.8cm, height=1cm, valign=c]{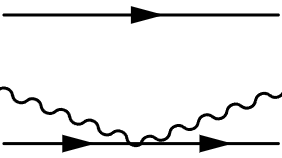}
        \caption{\label{fig:contact} Diagrams corresponding to contact configurations}
\end{figure}
\par
In case 3 all components will be given by contact evolution operators. Letting $\Pi_{\varsigma_1 \varsigma_2 \varsigma_3}$ be the projection operator of a three-body wave function onto its relevant component, we may write the evolution of each component in terms of the initial data as
\begin{equation}
\begin{split}
    \psi_{-\varsigma_1 \varsigma_2}(\textbf{x}_\text{ph},\textbf{x}_{\text{e}_1},\textbf{x}_{\text{e}_2})&=\Pi_{- \varsigma_1 \varsigma_2} \bigg( \big(C_{\theta_1}(\textbf{x}_\text{ph},\textbf{x}_{\text{e}_1}) \otimes  \mathcal{E}(\textbf{x}_{\text{e}_2}) \big) \mathring{\Psi}\bigg)
    \\
    \psi_{+\varsigma_1 \varsigma_2}(\textbf{x}_\text{ph},\textbf{x}_{\text{e}_1},\textbf{x}_{\text{e}_2})&=\Pi_{+ \varsigma_1 \varsigma_2} \bigg( \big(\mathcal{E}(\textbf{x}_{\text{e}_1}) \otimes C_{\theta_2}(\textbf{x}_\text{ph},\textbf{x}_{\text{e}_2})  \big) \mathring{\Psi}\bigg).
\end{split}
\end{equation}
\subsubsection{Coulomb Configurations}
In this section we will discuss the wave function on configurations where the electron light cones are allowed to intersect. These configurations are of great interest to us because enough time has passed for the photon to bounce from one electron to the other, which may produce an effective interaction between the two electrons. 
\par
Unfortunately, allowing the photon to bounce off one electron and then the other is not without consequence. When drawing the diagrams associated to these configurations, we must remember the rules: the wave function at $(\textbf{x}_\text{ph},\textbf{x}_{\text{e}_1},\textbf{x}_{\text{e}_2})$ is calculated by summing over all the possible diagrams which end in that state. For "Free" configurations there is only one type of diagram: free evolution of all three particles. On "Compton" configurations we were limited to diagrams with only one bounce. However, when we open the door to two bounces, we are in fact opening the door to any number of bounces. In the time it takes for two bounces to occur, so too could have three bounces, or four, etc, and so our diagram must be a super-position of all such interactions, see Fig.~\ref{fig:Infinite Diagram Sum}.
\begin{figure}[htbp]
\centering
\includegraphics[width=1.8cm, height=1cm, valign=c]{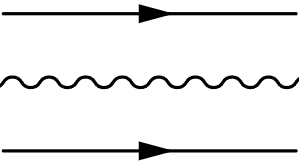} \quad
        $\mathlarger{+} \quad  e^{i\theta_2} $
        \includegraphics[width=1.78cm, height=1cm, valign=c]{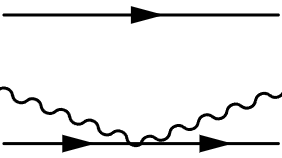}
        \quad 
        $\mathlarger{+} \quad e^{i(\theta_1+\theta_2)} $
        \includegraphics[width=1.8cm, height=0.9cm, valign=c]{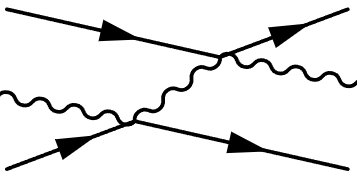}
        \quad
        $\mathlarger{+} \quad  e^{i(\theta_1+2\theta_2)} $
        \includegraphics[width=1.8cm, height=0.9cm, valign=c]{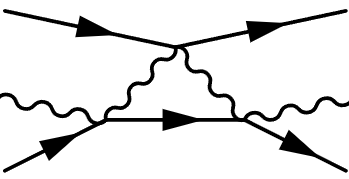}
       \\[10pt] $\mathlarger{+} $
% \end{figure}
% \begin{figure}[htbp]
% \centering
$e^{i2(\theta_1+\theta_2)} $
        \includegraphics[width=1.8cm, height=0.9cm, valign=c]{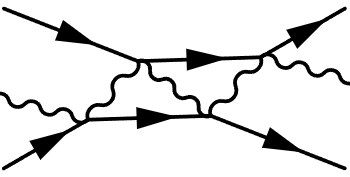} \quad
        $\mathlarger{+} \quad e^{i(2\theta_1+3\theta_2)} $
        \includegraphics[width=1.8cm, height=0.9cm, valign=c]{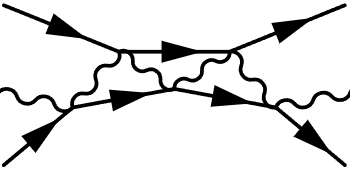} \quad
        $\mathlarger{+} \quad e^{i3(\theta_1+\theta_2)} $
        \includegraphics[width=1.8cm, height=0.9cm, valign=c]{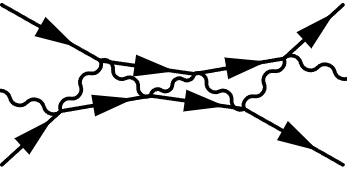}
        $\mathlarger{+} ... $
        \caption{The photon could have bounced arbitrarily many times between the electrons before the three propagated to their final destinations.}\label{fig:Infinite Diagram Sum}
\end{figure}
\par
This paints a bleak picture for our endeavor, as it is not clear whether the infinite sum of diagrams converges in some appropriate function space. Even if it does converge, it is not guaranteed that this limit is equal to the solution proven to exist in the Appendix.

\subsection{Convergence of Infinite Diagram Sum}
\subsubsection{Motivation for altered evolution}
In this section we prove convergence of the infinite sum of diagrams of Figure (\ref{fig:Infinite Diagram Sum}) which represents the wave function on "Coulomb" configurations. We consider a slightly altered evolution for our three-body wave function which approaches the original evolution under an appropriate limit. We will find it convenient to set up the altered evolution in the equal time setting $t_\text{ph}=t_{\text{e}_1}=t_{\text{e}_2}$, and later remark why an equal time convergence is enough to construct the wave function for "Coulomb" configurations in the multi-time setting. 
\par
We rectify the infinite bouncing problem by setting the boundary conditions to $0$ on configurations where $|s_{\text{e}_1} -s_{\text{e}_2}|<\epsilon$. By finite speed of propagation in configuration space, this has the desired effect of placing a lower bound of $\frac{\epsilon}{2}$ on the time it takes for data to bounce off one boundary and propagate to the other.
\par
In the language of Feynman diagrams, this altered evolution is equivalent to deleting all terms from our sum in which a photon bounced off an electron while the two electrons were too close. Over any finite time interval of length $T$, this will truncate our sum by deleting all terms which depict more than $N_\epsilon=1+\frac{2T}{\epsilon}$ many photon-electron collisions. By linearity, taking $\epsilon$ smaller re-introduces the deleted diagrams into our sum. Our goal in this section will be to prove the convergence of these "leaky boundary" evolutions in the limit as $\epsilon$ goes to $0$ and show that this limit returns the correct dynamics.
\subsubsection{Leaky Boundary IBVP}
In this section we will prove convergence of the "leaky boundary" evolution. This will be done in the equal time setting, where $t_\text{ph}=t_{\text{e}_1}=t_{\text{e}_2}$. The Hamiltonian of our system is
\begin{equation}\label{Three Part Hamiltonian}
    \hat{H}=\big(\hat{H}_\text{ph} \otimes \mathbb{1} \otimes \mathbb{1} \big) +  \big(\mathbb{1} \otimes \hat{H}_{\text{e}_1} \otimes \mathbb{1}\big)+ \big(\mathbb{1} \otimes \mathbb{1} \otimes \hat{H}_{\text{e}_2}\big), \quad D(\hat{H})=C_c^\infty(\mathcal{S}_1, \mathbb{C}^8)
\end{equation}
where 
\begin{equation}
    \hat{H}_\text{ph}= -i\hbar\gamma^0\gamma^1\partial_{s_\text{ph}}, \quad \hat{H}_{\text{e}_1}=\gamma^0(m_\text{e}-i\hbar\gamma^1\partial_{s_{\text{e}_1}}) \quad \hat{H}_{\text{e}_2}=\gamma^0(m_\text{e}-i\hbar\gamma^1\partial_{s_{\text{e}_2}}).
\end{equation}
The equal time dynamics will take place on a subset of $\mathbb{R}_{\text{ph}}\times \mathbb{R}_{\text{e}_1}\times\mathbb{R}_{\text{e}_2}$:
\begin{align}
&\mathcal{S}_1\coloneqq\left\{(s_{\text{ph}},s_{\text{e}_1},s_{\text{e}_2}):s_{\text{e}_1}<s_{\text{ph}}<s_{\text{e}_2}\right\},
\end{align}
with boundary conditions that are linear relations between limits along events in $\mathcal{S}_1$ towards the coincidence sets $\mathcal{C}_1$ and $\mathcal{C}_2$, defined by
\begin{equation}
\begin{split}
\mathcal{C}_1 &\coloneqq \left\{(s_{\text{ph}},s_{\text{e}_1},s_{\text{e}_2})\in \partial\mathcal{S}_1:s_{\text{e}_1}=s_{\text{ph}} \right\}\\
\mathcal{C}_2 &\coloneqq \left\{(s_{\text{ph}},s_{\text{e}_1},s_{\text{e}_2})\in \partial\mathcal{S}_1:s_{\text{e}_2}=s_{\text{ph}} \right\}.
\end{split}
\end{equation}
\begin{theorem}\label{thm:Leaky}
 Let $\theta_1,\theta_2 \in [0,2\pi)$ be two constant phases, and let $\mu_\epsilon (s):\mathbb{R^+}\rightarrow [0,1]$ be a smooth transition function which is zero on $[0,\epsilon]$ and 1 on $[2\epsilon,\infty)$. Let  the function $\mathring \Psi : \mathcal{S}_1 \rightarrow \mathbb{C}^8$ be bounded and compactly supported in ${\mathcal{S}_1}$. Then the following initial-boundary value problem for the wave function $\Psi: \mathbb{R}^{+}\times \mathcal{S}_1 \rightarrow \mathbb{C}^8$,
\begin{equation}\label{Leaky IVP}
\left\{
\begin{array}{rclr} 
i\hbar \partial_t \Psi &= &\hat{H}\Psi  &\text{in } \mathbb{R}^+ \times \mathcal{S}_1 
\\
\Psi \big{|}_{t=0} &= &\mathring \Psi   
\\
\psi_{-+\varsigma_2} &= &\mu_\epsilon(|s_{\text{e}_1}-s_{\text{e}_2}|)e^{i\theta_1}\psi_{+-\varsigma_2} \hspace{0.5cm}&\text{on } \mathcal{C}_1
\\
\psi_{+\varsigma_1-} &= &\mu_\epsilon(|s_{\text{e}_1}-s_{\text{e}_2}|)e^{i\theta_2}\psi_{-\varsigma_1+}  &\text{on } \mathcal{C}_2
\end{array}
\right.
\end{equation}
has a unique global-in-time solution that is supported in $\overline{\mathcal{S}_1}$, and depends continuously on the initial data $\mathring \Psi$.
\end{theorem}

\begin{proof}
By construction, these boundary conditions ensure that data which bounces off one boundary $\mathcal{C}_1$ takes at least $\frac{\epsilon}{2}$ time to propagate to the other boundary $\mathcal{C}_2$ and vice versa. Because of this, we will find it easier to only solve for the evolution to the leaky boundary IBVP for all times up to $\frac{\epsilon}{2}$. Since $\epsilon$ is fixed, we may repeatedly apply the evolution operator with time step size $\frac{\epsilon}{2}$ until any desired time is reached. 

%As we construct our solution, we also prove the estimate \eqref{est:Linf} for the growth of $\|\Psi\|_\infty$.

\par
Fix $t\leq \frac{\epsilon}{2}$. We wish to solve for $\Psi(t)$ on $\mathcal{S}_1$. Similarly to previous interacting systems, we approach this by splitting up our configuration space into regions for which the particle configurations are "Free" or "Compton". 
\begin{equation}
    \mathcal{S}_1= \mathcal{F} \cup \mathcal{X},
\end{equation}
where
\begin{equation}
    \mathcal{F}:=\left \{(s_\text{ph},s_{\text{e}_1},s_{\text{e}_2}) \in \mathcal{S}_1 : s_{\text{e}_1}+t< s_\text{ph}-t, \quad  s_\text{ph}+t< s_{\text{e}_2}-t \right\},
\end{equation}
\begin{equation}
    \mathcal{X}:=\mathcal{S}_1\setminus\mathcal{F}= \left\{(s_\text{ph},s_{\text{e}_1},s_{\text{e}_2}) \in \mathcal{S}_1 : s_{\text{e}_1}+t\geq s_\text{ph}-t \ \mbox{ or } \quad  s_\text{ph}+t\geq s_{\text{e}_2}-t \right\}.
\end{equation}
\paragraph{Proof for free configurations:}
Similarly to the interacting multi-time problems encountered earlier, the solution for free configurations is simply given by the free three body evolution operator
\begin{equation}\label{Leaky Free}
    \Psi(t,s_\text{ph},s_{\text{e}_1}, s_{\text{e}_2})=\big(\mathcal{P}(\textbf{x}_\text{ph})\otimes \mathcal{E}(\textbf{x}_{\text{e}_1}) \otimes  \mathcal{E}(\textbf{x}_{\text{e}_2}) \big) \mathring{\Psi},
\end{equation}
where here $\textbf{x}_\text{ph}=(t,s_\text{ph})$, $\textbf{x}_{\text{e}_1}=(t,s_{\text{e}_1})$ and analogously for $\textbf{x}_{\text{e}_2}$. 

% By (\ref{Dirac Estimate}), we have that for all Free configurations
% \begin{equation}\label{Free Estimate}
%     |\psi_{\varsigma_1 \varsigma_2 \varsigma_2}(t,s_\text{ph},s_{\text{e}_1}, s_{\text{e}_2})|\leq (1+\omega^2 t +\omega t)^2||\mathring{\Psi}||_\infty \leq (1+A_0t)||\mathring{\Psi}||_\infty,
% \end{equation}
% for some $A_0$, where here we've used that $t<1$ to absorb all higher order terms.
\paragraph{Proof for Compton configurations:}
We begin by splitting up the Compton Configurations into three cases:

\textbf{Case 1:} This case covers all configurations for which the first electron is interacting with the photon, while the second electron is freely evolving. In this regime, the $\varsigma_0=+$ components are freely evolving, and so they are given by
\begin{equation}\label{Leaky Compton Free}
    \psi_{+\varsigma_1 \varsigma_2}=\Pi_{+\varsigma_1 \varsigma_2}\bigg( \big(\mathcal{P}(\textbf{x}_\text{ph})\otimes \mathcal{E}(\textbf{x}_{\text{e}_1}) \otimes  \mathcal{E}(\textbf{x}_{\text{e}_2}) \big) \mathring{\Psi}\bigg),
\end{equation}
where $\Pi$ is the projection operator onto the associated component. 
%It follows that the estimate (\ref{Free Estimate}) holds for these components as well. 
The $\varsigma_0=-$ components require a different approach.
\par
We begin by expanding the first equation of (\ref{Leaky IVP}) in terms of the components of $\Psi$. This returns sourced transports 
\begin{equation}\label{equal time sourced transports}
    (\partial_t - \varsigma_0 \partial_{s_{\text{ph}}} - \varsigma_1 \partial_{s_{\text{e}_1}} - \varsigma_2 \partial_{s_{\text{e}_2}})\psi_{\varsigma_0 \varsigma_1 \varsigma_2}+i\omega \psi_{\varsigma_0 \overline{\varsigma}_1 \varsigma_2}+i\omega\psi_{\varsigma_0 {\varsigma}_1 \overline{\varsigma}_2}=0.
\end{equation}
Applying the operator $(\partial_t - \varsigma_0 \partial_{s_{\text{ph}}} - \overline{\varsigma}_1 \partial_{s_{\text{e}_1}} - \varsigma_2 \partial_{s_{\text{e}_2}})$ to equation (\ref{equal time sourced transports}) returns
\begin{equation}\label{Equal time KG}
    ((\partial_{\nu} - \varsigma_2 \partial_{s_{\text{e}_2}})^2-\partial_{s_{\text{e}_1}}^2)\psi_{-+\varsigma_2} +2i\omega \partial_{\nu}\psi_{-+ \overline{\varsigma}_2}=0,
\end{equation}
where $\partial_\nu=\partial_t +\partial_{s_\text{ph}}$. It is then easy to verify that the following solution formula for $\psi_{-+\varsigma_2}$ satisfies equation (\ref{Equal time KG}) in the space of Compton configurations. 
\begin{equation}\label{Leaky Gour Solution}
    \psi_{-+\varsigma_2}(t, s_\text{ph},s_{\text{e}_1},s_{\text{e}_2})=G^R(t,s_{\text{e}_1}-(s_\text{ph}-t))F_{\varsigma_2} + G^L(t,s_{\text{e}_1}-(s_\text{ph}-t))G_{\varsigma_2},
\end{equation}
where 
\begin{equation}\label{Leaky Gour Data 1}
\begin{split}
G_{\varsigma_2}(c)&=\Pi_{\varsigma_2}\bigg( \mathcal{E}(t-c,s_{\text{e}_2})\begin{pmatrix}
        \psi_{-++}
        \\
        \psi_{-+-}
    \end{pmatrix}
    (c,\textbf{c},s_{\text{e}_2}) \bigg)
    =\psi_{-+\varsigma_2}(c,\textbf{c},s_{\text{e}_2}+\varsigma_2 (t-c)) 
    \\
    &-\frac{\omega}{2}\int_{s_{\text{e}_2}-(t-c)}^{s_{\text{e}_2}+(t-c)} J_{1}(\omega\sqrt{(t-c)^2-(s_{\text{e}_2}-\sigma)^2})\frac{\sqrt{(t-c)-\varsigma(s_{\text{e}_2}-\sigma)}}{\sqrt{(t-c)+\varsigma(s_{\text{e}_2}-\sigma)}}\psi_{-+\varsigma_2}(c,\textbf{c},\sigma) d\sigma
    \\
    &
    -\frac{i\omega}{2}\int_{s_{\text{e}_2}-(t-c)}^{s_{\text{e}_2}+(t-c)} J_{0}(\omega\sqrt{(t-c)^2-(s_{\text{e}_2}-\sigma)^2}) \psi_{-+\overline{\varsigma}_2}(c,\textbf{c},\sigma) d\sigma,
\end{split}
\end{equation}
and 
$\textbf{c}=(s_\text{ph}-t+c,s_\text{ph}-t-c) \in \mathbb{R}_\text{ph}\times \mathbb{R}_{\text{e}_1}$. We define $F_{\varsigma_2}$ similarly, but taking care to correctly apply our boundary condition
\begin{equation}\label{Leaky Gour Data 2}
    F_{\varsigma_2}(b)= \Pi_{\varsigma_2}\bigg( \mathcal{E}(t-b,s_{\text{e}_2})\begin{pmatrix}
        \mu_\epsilon(|s_{\text{e}_1}-s_{\text{e}_2}|)e^{i\theta_1}\psi_{+-+}
        \\
        \mu_\epsilon(|s_{\text{e}_1}-s_{\text{e}_2}|)e^{i\theta_1}\psi_{+--}
    \end{pmatrix}
    (b,\textbf{b},s_{\text{e}_2}) \bigg),
\end{equation}
To solve for the remaining two components, we apply the operator $(\partial_t - \varsigma_0 \partial_{s_{\text{ph}}} - {\varsigma}_1 \partial_{s_{\text{e}_1}} - \overline{\varsigma}_2 \partial_{s_{\text{e}_2}})$ to equation (\ref{equal time sourced transports}) and find that $\psi_{--\varsigma_2}$ satisfies
\begin{equation}\label{Equal time Sourced KG}
    (\partial_u^2 - \partial_{s_{\text{e}_2}}^2 +\omega^2)\psi_{--\varsigma_2} +i\omega (\partial_u + \overline{\varsigma}_2\partial_{\text{e}_2})\psi_{-+\varsigma_2}+\omega^2\psi_{-+\overline{\varsigma}_2}=0,
\end{equation}
where $\partial_u=\partial_t+\partial_{s_\text{ph}}+\partial_{s_{\text{e}_1}}$. Equation (\ref{Equal time Sourced KG}) is an inhomogenous Klein-Gordon equation with sources that are known to us in terms of the initial data from equation (\ref{Leaky Gour Solution}). We solve for the two remaining components by setting up and solving the initial value problem
\begin{equation}
    \begin{cases}
    (\partial_u^2 - \partial_{s_{\text{e}_2}}^2 +\omega^2)\psi_{--\varsigma_2} +i\omega (\partial_u + \overline{\varsigma}_2\partial_{\text{e}_2})\psi_{-+\varsigma_2}+\omega^2\psi_{-+\overline{\varsigma}_2}=0
    \\
    \psi_{--\varsigma_2}\bigg{|}_{u=0}=\mathring{\psi}_{--\varsigma_2}(s_\text{ph}-t,s_{\text{e}_1}-t,s_{\text{e}_2})
    \\
    \partial_u\psi_{--\varsigma_2}\bigg{|}_{u=0}=\bigg(\varsigma_2 \partial_{\text{e}_2}\mathring{\psi}_{--\varsigma_2} +i\omega \mathring{\psi}_{-+\varsigma_2} +i\omega \mathring{\psi}_{--\overline{\varsigma}_2}\bigg)(s_\text{ph}-t,s_{\text{e}_1}-t,s_{\text{e}_2}).
\end{cases}
\end{equation}
By linearity, we can write the solution for $\psi_{--\varsigma_2}$ as $U_{\varsigma_2} +S_{\varsigma_2}$, where $U_{\varsigma_2}$ represents the freely evolving initial data  and $S_{\varsigma_2}$ the sourced data. These solutions are given as
\begin{eqnarray}
 U_{\varsigma_2}(t,s_\text{ph},s_{\text{e}_1},s_{\text{e}_2})&=&\mathring{\psi}_{--\varsigma_2}(s_\text{ph}-t,s_{\text{e}_1}-t,s_{\text{e}_2}+\varsigma_2 t) 
\nonumber\\
&&\mbox{}-\frac{\omega}{2}\int_{s_{\text{e}_2}-t}^{s_{\text{e}_2}+t} J_{1}(\omega\sqrt{t^2-(s_{\text{e}_2}-\sigma)^2})\frac{\sqrt{t-\varsigma(s_{\text{e}_2}-\sigma)}}{\sqrt{t+\varsigma_2(s_{\text{e}_2}-\sigma)}}\mathring{\psi}_{--\varsigma}(s_\text{ph}-t,s_{\text{e}_1}-t,\sigma) d\sigma
\\
&&\mbox{}-\frac{i\omega}{2}\int_{s_{\text{e}_2}-t}^{s_{\text{e}_2}+t} J_{0}(\omega\sqrt{t^2-(s_{\text{e}_2}-\sigma)^2}) \big(\mathring{\psi}_{-+\varsigma_2}+\mathring{\psi}_{--\overline{\varsigma}_2} \big)(s_\text{ph}-t,s_{\text{e}_1}-t,\sigma) d\sigma
\nonumber
\end{eqnarray}
and
\begin{eqnarray}
    S_{\varsigma_2}(t,s_\text{ph},s_{\text{e}_1},s_{\text{e}_2})&=&
    -\frac{\omega^2}{2}\int_0^t \int_{s_{\text{e}_2}-(t-u)}^{s_{\text{e}_2}+(t-u)} J_{0}(\omega\sqrt{(t-u)^2-(s_{\text{e}_2}-\sigma)^2}) {\psi}_{-+\overline{\varsigma}_2}(u,s_\text{ph}-t+u,s_{\text{e}_1}-t+u,\sigma) d\sigma du
\\
    &-&\frac{i\omega}{2}\int_0^t \int_{s_{\text{e}_2}-(t-u)}^{s_{\text{e}_2}+(t-u)} J_{0}(\omega\sqrt{(t-u)^2-(s_{\text{e}_2}-\sigma)^2}) (\partial_u -\overline{\varsigma}_2\partial_{\text{e}_2}){\psi}_{-+{\varsigma}_2}(u,s_\text{ph}-t+u,s_{\text{e}_1}-t+u,\sigma) d\sigma du,\nonumber
\end{eqnarray}
% This concludes our construction of $\Psi$ for this case of Compton configurations. Before moving on to the next case of configurations, we briefly note that
% \begin{equation}
%     |U_{\varsigma_2}(t,s_\text{ph},s_{\text{e}_1},s_{\text{e}_2})|\leq (1+\omega^2 t + 2\omega t)||\mathring{\Psi}||_\infty,
% \end{equation}
% while, after integrating by parts and using the estimate from equation (\ref{Gour Estimate}), we find
% $S_{\varsigma_2}(t,s_\text{ph},s_{\text{e}_1},s_{\text{e}_2})\leq O(t)||\mathring{\Psi}||_\infty$. This, in combination with our previous estimates imply that all $8$ components of $\Psi$ satisfy $|\psi_{\varsigma_0 \varsigma_1 \varsigma_2}| \leq (1+O(t))||\mathring{\Psi}||_\infty$ for this case of Compton configurations.
\newline
\textbf{Case 2:} This case covers all configurations for which the second electron is interacting with the photon, while the first electron is freely evolving. In this regime, the $\varsigma_0=-$ components are freely evolving, and their solution formulas are given by
\begin{equation}\label{Leaky Compton C2 Free}
    \psi_{-\varsigma_1 \varsigma_2}=\Pi_{-\varsigma_1 \varsigma_2}\bigg( \big(\mathcal{P}(\textbf{x}_\text{ph})\otimes \mathcal{E}(\textbf{x}_{\text{e}_1}) \otimes  \mathcal{E}(\textbf{x}_{\text{e}_2}) \big) \mathring{\Psi}\bigg)
\end{equation}
Meanwhile, the $\varsigma_0=+$ components will be given by Goursat type evolutions for $\psi_{+\varsigma_1 -}$, which is then plugged into an inhomogenous Klein-Gordon equation for $\psi_{+\varsigma_1 +}$. 
% The construction of these solutions are nearly identical to the construction for Case 1, with the desired estimates from Case 1 holding in this regime as well.
\newline
\textbf{Case 3:} For the final set of configurations we will construct $\Psi$ on are those for which both electrons are interacting with the photon. This is similar to the third case discussed in the multi-time setting, except now we are not excluding configurations where the two electrons are interacting, previously called Coulomb configurations. Since we are only evolving for this short time, we will find that the solution formulas constructed for this case also extends to Coulomb configurations identically. In this regime, non of the wave function components are freely evolving, their solutions will be given by Goursat evolutions and inhomogenous Klein-Gordon equations which are identical to those found in the previous two cases. 
\par
For example, the solution formula for $\psi_{-+\varsigma_2}$ in this regime is still given by (\ref{Leaky Gour Solution}). To see that this solution formula makes sense, we note that the data necessary to define $G_{\varsigma_2}$ as in (\ref{Leaky Gour Data 1}) is data for which the first electron is freely evolving, while the second electron interacts with the photon. This data for this is given by equation (\ref{Leaky Compton C2 Free}). Of greater interest is the data necessary to define $F_{\varsigma_2}$ as in (\ref{Leaky Gour Data 2}), which written explicitly is
\begin{equation}
\begin{split}
F_{\varsigma_2}(b)&= \Pi_{\varsigma_2}\bigg( \mathcal{E}(t-b,s_{\text{e}_2})\begin{pmatrix}
        \mu_\epsilon(|s_{\text{e}_1}-s_{\text{e}_2}|)e^{i\theta_1}\psi_{+-+}
        \\
        \mu_\epsilon(|s_{\text{e}_1}-s_{\text{e}_2}|)e^{i\theta_1}\psi_{+--}
    \end{pmatrix}
    (b,\textbf{b},s_{\text{e}_2}) \bigg)=\big(\mu_\epsilon e^{i\theta_1}\psi_{+-\varsigma_2}\big)(b,\textbf{b},s_{\text{e}_2}+\varsigma_2 (t-b)) 
    \\
    &-\frac{\omega}{2}\int_{L_\epsilon} J_{1}(\omega\sqrt{(t-b)^2-(s_{\text{e}_2}-\sigma)^2})\frac{\sqrt{(t-b)-\varsigma_2(s_{\text{e}_2}-\sigma)}}{\sqrt{(t-b)+\varsigma_2(s_{\text{e}_2}-\sigma)}}\big(\mu_\epsilon e^{i\theta_1}\psi_{+-\varsigma_2}\big)(b,\textbf{b},\sigma) d\sigma
    \\
    &
    -\frac{i\omega}{2}\int_{L_\epsilon} J_{0}(\omega\sqrt{(t-b)^2-(s_{\text{e}_2}-\sigma)^2}) \big(\mu_\epsilon e^{i\theta_1}\psi_{+-\overline{\varsigma}_2}\big)(b,\textbf{b},\sigma) d\sigma,
\end{split}
\end{equation}
where $\textbf{b}=(s_\text{ph}-t+b,s_\text{ph}-t+b)\in \mathbb{R}_\text{ph} \times \mathbb{R}_{\text{e}_1}$, and 
\begin{equation}
    L_\epsilon:=\left\{\sigma \in \mathbb{R}_{\text{e}_2}: |\sigma - s_{\text{e}_2}|<t-b, \quad |\sigma - (s_\text{ph}-t +b)|\geq \epsilon  \right\}.
\end{equation}
This extra condition on the region of integration comes from the lack of support of $\mu_\epsilon$ for configurations where the two electrons are $\epsilon$ distance away. 
 We find that only freely evolving data is necessary to define $F_{\varsigma_2}$ even on configurations where the two electrons are interacting, since we are only evolving up to $t\leq \frac{\epsilon}{2}$. It follows that for all $t\leq \frac{\epsilon}{2}$, we have constructed a global solution for $\psi_{-+\varsigma_2}$ in $\mathcal{S}_1$ using solution formulas (\ref{Leaky Free}) for free configurations, (\ref{Leaky Compton C2 Free}) for Compton configurations of case 2, and  (\ref{Leaky Gour Solution}) for Compton configuration cases 1 and 3. The remaining six components can be treated similarly, with the formula $\psi_{+\varsigma_1 -}$ also given in terms of Goursat evolution operators while the other four components are solved by studying inhomogeneous Klein-Gordon equations. 
 %Since all of these formulas are identical to the ones given in case 1, it can be shown that all of the desired estimates hold in this regime as well.
 This concludes our construction of the small time solution to the leaky boundary IBVP, and thus the proof of Theorem~\ref{thm:Leaky}.
\end{proof}
\subsubsection{Main Result}
For $\mathring{\Psi}$ bounded and compactly supported in $\mathcal{S}_1$, let $U_\epsilon(t)$ denote the map $\mathring{\Psi}\to \Psi_\epsilon(t)$, where $\Psi_\epsilon(t)$ solves the IBVP (\ref{Leaky IVP}) with initial data $\mathring{\Psi}$. We will now state and prove the main result of this paper.
\begin{theorem}\label{thm:main}
    Let $\mathring{\Psi}:\mathcal{S}_1 \to \mathbb{C}^8$ bounded and compactly supported. Then taking a limit of the solutions $\Psi_\epsilon(t)$ as $\epsilon \to 0$ returns the equal-time dynamics we originally set out to solve for, i.e.  $\lim_{\epsilon \to 0}U_\epsilon(t)\mathring{\Psi}=e^{\frac{-it}{\hbar} \tilde{H}}\mathring{\Psi}$ where $\tilde{H}$ is the self-adjoint extension of $\hat{H}$ with domain
   \begin{equation}\label{SA dom}
    D(\tilde{H}):=\left\{\Psi \in D(\hat{H}^*): \psi_{-+\varsigma_2} = e^{i\theta_1}\psi_{+-\varsigma_2} \text{ on } \mathcal{C}_1, \quad \psi_{+\varsigma_1-} = e^{i\theta_2}\psi_{-\varsigma_1+}  \text{ on } \mathcal{C}_2 \right\}.
\end{equation}
Hence the Feynman diagram sum depicted in Fig.~\ref{fig:Infinite Diagram Sum} converges and produces the correct dynamics.
\end{theorem}
The following is an immediate corollary of the above theorem:
\begin{corollary}
    Let $\mathring{\Psi}:\mathcal{S}_1 \to \mathbb{C}^8$ be $C^1$ and compactly supported in $\mathcal{S}_1$. Then $\lim_{\epsilon \to 0} \Psi_\epsilon$ has a representative in $L^2(\mathcal{S}_1)$ which is continuous with bounded derivatives. 
\end{corollary}
\begin{proof}
In Theorem~\ref{thm:2body}
we showed that there exists a solution to the multi-time IBVP that is continuous and with bounded derivatives. This solution must be unique by probability conservation. Restricting this solution to the equal-time setting must therefore be a representative of  $e^{\frac{-it}{\hbar}\tilde{H}} \mathring{\Psi}$.
\end{proof}
Here is an outline of the proof of Theorem~\ref{thm:main}.
\begin{itemize}
    \item[STEP 1.] We first show that for each $\epsilon>0$, $U_\epsilon(t)$ extends to a $C_0$ contraction semigroup on $L^2(\mathcal{S}_1)$, which by the theorem of Lumer-Phillips implies $U_\epsilon(t)=e^{\frac{-it}{\hbar}\hat{H}_\epsilon}$ for some maximally dissipative extension $\frac{-i}{\hbar}\hat{H}_\epsilon$ of $\frac{-i}{\hbar}\hat{H}$.
    \item[STEP 2.] We next review the approximation theorems of Kurtz, and give sufficient conditions on the generators $\frac{-i}{\hbar}\hat{H}_\epsilon$ for the evolution mappings $e^{\frac{-it}{\hbar}\hat{H}_\epsilon}$ to converge to another $C_0$ contraction semigroup as $\epsilon \to 0$.
    \item[STEP 3.] We then write down explicit expressions for the generators $\frac{-i}{\hbar}\hat{H}_\epsilon$, prove $\tilde{H}$ is self-adjoint, and conclude by applying the approximation theorems to show $\lim_{\epsilon \to 0}e^{\frac{-it}{\hbar}\hat{H}_\epsilon}=e^{\frac{-it}{\hbar}\tilde{H}}$.
\end{itemize}
\noindent\paragraph{Step 1 in the proof of Thm.~\ref{thm:main}:} We first review the definition of a $C_0$ contraction semigroup.
\begin{definition}
    A $\mathbf{C_0}$ \textbf{contraction semigroup} on $L^2(\mathcal{S}_1)$ is a one-parameter family of linear maps $U(t):L^2(\mathcal{S}_1) \to L^2(\mathcal{S}_1)$ defined for $t\geq 0$ with the following properties
    \begin{enumerate}
        \item[(a)] The maps form a semigroup under composition, i.e $U(t) U(s)= U(t+s)$ for all $t,s \geq 0$, with $U(0)=\mathbb{1}$.
        \item[(b)] They are strongly continuous,   $\lim_{t \to t_0}||U(t) \Psi - U(t_0)\Psi||_{L^2(\mathcal{S}_1)}=0$ for all $\Psi \in L^2(\mathcal{S}_1)$, $t_0 \geq 0$. 
        \item[(c)] They are contractions, $||U(t) \Psi||_{L^2(\mathcal{S}_1)}\leq ||\Psi||_{L^2(\mathcal{S}_1)}$ for all $\Psi \in L^2(\mathcal{S}_1)$, $t_0 \geq 0$.
    \end{enumerate}
\end{definition}
\begin{proposition}
    For each $\epsilon>0$, the family of solution maps $U_\epsilon(t)$ extends to a $C_0$ contraction semigroup on $L^2(\mathcal{S}_1)$.
\end{proposition}
\begin{proof}
    Fix $\epsilon>0$. Linearity of the solution map $U_\epsilon(t)$ follows directly from linearity of the IBVP (\ref{Leaky IVP}).\begin{comment}
    , since for any $\mathring{\Psi},\mathring{\Phi}$ bounded and compactly supported in $\mathcal{S}_1$, the unique solution to (\ref{Leaky IVP}) with initial data $\alpha \mathring{\Psi}+\beta \mathring{\Phi}$ is $\alpha U_\epsilon(t)\mathring{\Psi}+ \beta U_\epsilon(t) \mathring{\Phi}$ .
    \end{comment} 
     We now show the linear maps are contractive in $L^2$ norm when acting on smooth compactly supported functions. Let $\mathring{\Psi} \in C_c^\infty(\mathcal{S}_1)$, and $\Psi(t):=U_\epsilon(t) \mathring{\Psi}$ satisfy (\ref{Leaky IVP}) with initial data $\mathring{\Psi}$. The evolution equation admits a probability current $j^\mu_\Psi$ defined by 
\begin{align}
    &j^0_\Psi:=\sum_{\varsigma_0 ,\varsigma_1 ,\varsigma_2} |\psi_{\varsigma_0 \varsigma_1 \varsigma_2}|^2, \quad  
    &j^1_\Psi := \sum_{ \varsigma_1 ,\varsigma_2} |\psi_{- \varsigma_1 \varsigma_2}|^2 -|\psi_{+ \varsigma_1 \varsigma_2}|^2, 
    \\ 
    &j^2_\Psi  :=\sum_{\varsigma_0, \varsigma_2} |\psi_{\varsigma_0 -\varsigma_2}|^2-|\psi_{\varsigma_0 + \varsigma_2}|^2, \quad 
    &j^3_\Psi  :=\sum_{\varsigma_0, \varsigma_1} |\psi_{\varsigma_0 \varsigma_1 -}|^2-|\psi_{\varsigma_0 \varsigma_1 +}|^2 
\end{align}
which satisfies a continuity equation
\begin{equation}\label{Cont 1}
    \frac{\partial j^0_\Psi}{\partial t}=-\vec{\nabla}\cdot \vec{j_\Psi} \quad \text{on }\mathcal{S}_1.
\end{equation}
Integrating this over $\mathcal{S}_1$ returns
\begin{equation}\label{Cont 2}
    \frac{d}{dt} ||\Psi||_{L^2(\mathcal{S}_1)}=\frac{d}{dt} ||j^0_\Psi||_{L^1(\mathcal{S}_1)}= \int_{\mathcal{C}^1 \cup \mathcal{C}^2}-\vec{j}_\Psi\cdot \vec{n} dS,
\end{equation}
where $dS$ is the surface element of $\mathcal{C}^1 \cup \mathcal{C}^2$ and $\vec{n}$ its normals. We compute that
\begin{equation}\label{Cont 3}
    -\vec{j}_\Psi\cdot \vec{n}=\frac{1}{\sqrt{2}}\begin{cases}
    j^1_\Psi-j^2_\Psi=|\psi_{-+-}|^2+|\psi_{-++}|^2-|\psi_{+--}|^2-|\psi_{+-+}|^2 \text{ on } \mathcal{C}^1
    \\
    j^3_\Psi-j^1_\Psi=|\psi_{+--}|^2+|\psi_{++-}|^2-|\psi_{--+}|^2-|\psi_{-++}|^2 \text{ on } \mathcal{C}^2.
    \end{cases}
\end{equation}
Since $\Psi$ satisfies the boundary conditions of (\ref{Leaky IVP}), this reduces to
\begin{equation}\label{Cont 4}
    -\vec{j}_\Psi\cdot \vec{n}=\frac{1}{\sqrt{2}}
    \begin{cases}
    ((\mu_\epsilon)^2 - 1)\big(|\psi_{+--}|^2+|\psi_{+-+}|^2\big) \text{ on } \mathcal{C}_1 \\
    ((\mu_\epsilon)^2 - 1)\big(|\psi_{--+}|^2+|\psi_{-++}|^2\big) \text{ on } \mathcal{C}_2.
    \end{cases}
\end{equation}
\begin{comment}
It follows that the flux of the probability current is zero through $\mathcal{C}^1 \cup \mathcal{C}^2$ except for configurations where $|s_{\text{e}_2}-s_{\text{e}_1}|<2\epsilon$. We will call these regions $\mathcal{C}_1^\epsilon$ and $\mathcal{C}_2^\epsilon$, so that
\begin{equation}
    \frac{d}{dt}||\rho||_1 =\int_{\mathcal{C}_1^\epsilon \cup \mathcal{C}_2^\epsilon} -\vec{j}\cdot \vec{n} dS
\end{equation}
\end{comment}
The probability current for $\Psi$ satisfies $-\vec{j}_\Psi\cdot \vec{n}\leq 0$ everywhere on $\mathcal{C}_1 \cup \mathcal{C}_2$, hence $||U_\epsilon(t) \mathring{\Psi}||_{L^2(\mathcal{S}_1)}\leq ||\mathring{\Psi}||_{L^2(\mathcal{S}_1)}$ for all $t\geq 0$. The subspace $C_c^\infty(\mathcal{S}_1)$ is dense in $L^2(\mathcal{S}_1)$, so $U_\epsilon(t)$ extends to a one parameter family of linear contractions on $L^2(\mathcal{S}_1)$. 
\par
For smooth compactly supported initial data $\mathring{\Psi}$ and $s>0$, it is easy to verify that $U_\epsilon(t+s)\mathring{\Psi}$ solves the IBVP (\ref{Leaky IVP}) with initial data $U_\epsilon(s) \mathring{\Psi}$, so $U_\epsilon(t)U_\epsilon(s)\mathring{\Psi}=U_\epsilon(t+s)\mathring{\Psi}$. Hence the bounded linear operators $U_\epsilon(t)U_\epsilon(s)$ and $U_\epsilon(t+s)$ agree on a dense subspace of $L^2(\mathcal{S}_1)$, and therefore the semigroup property $U_\epsilon(t)U_\epsilon(s)=U_\epsilon(t+s)$ holds.
\par
To prove continuity in time, we will once again let $\mathring{\Psi}\in C_c^\infty(\mathcal{S}_1)$, and let $\Phi(t):=(U_\epsilon(t) - \mathbb{1})\mathring{\Psi}$. $\Phi$ satisfies the IBVP
 \begin{equation}
\left\{
\begin{array}{rclr} 
i\hbar \partial_t \Phi &= &\hat{H}(\Phi+\mathring{\Psi})  &\text{in } \mathbb{R}^+ \times \mathcal{S}_1 
\\
\Phi \big{|}_{t=0} &= &0
\\
\phi_{-+\varsigma_2} &= &\mu_\epsilon(|s_{\text{e}_1}-s_{\text{e}_2}|)e^{i\theta_1}\phi_{+-\varsigma_2} \hspace{0.5cm}&\text{on } \mathcal{C}_1
\\
\phi_{+\varsigma_1-} &= &\mu_\epsilon(|s_{\text{e}_1}-s_{\text{e}_2}|)e^{i\theta_2}\phi_{-\varsigma_1+}  &\text{on } \mathcal{C}_2
\end{array}
\right.
\end{equation}
Integrating the continuity equation for this IBVP returns
\begin{equation}
    \frac{d}{dt}||\Phi||^2_{L^2(\mathcal{S}_1)}=2\text{Re}\langle \frac{-i}{\hbar}\hat{H}\mathring{\Psi},\Phi \rangle_{L^2(\mathcal{S}_1)}+\int_{\mathcal{C}_1 \cup \mathcal{C}_2}-\vec{j}_\Phi \cdot \vec{n}dS 
\end{equation}
where $\vec{j}_\Phi$ is the probability current defined similarly for $\Phi$. This current also satisfies $\vec{j}_\Phi\cdot \vec{n}\leq 0$ everywhere on $\mathcal{C}_1 \cup \mathcal{C}_2$ since $\Phi$ satisfies the same leaky boundary conditions, therefore
\begin{equation}
    \frac{d}{dt}||\Phi||^2_{L^2(\mathcal{S}_1)}\leq 2\text{Re}\langle\frac{-i}{\hbar} \hat{H}\Psi,\Phi \rangle_{L^2(\mathcal{S}_1)}\leq \frac{2}{\hbar}||\hat{H}\mathring{\Psi}||_{L^2(\mathcal{S}_1)}||\Phi||_{L^2(\mathcal{S}_1)} \quad \Rightarrow \quad ||\Phi(t)||_{L^2(\mathcal{S}_1)}\leq \frac{t}{\hbar}||\hat{H}\mathring{\Psi}||_{L^2(\mathcal{S}_1)}.
\end{equation}
So $||(U_\epsilon(t)-\mathbb{1})\mathring{\Psi}||_{L^2(\mathcal{S}_1)}\to 0$ as $t \to 0$. To show this convergence holds for any $\mathring{\Psi}\in L^2(\mathcal{S}_1)$, let $\mathring{\Psi}_n \in C_c^\infty(\mathcal{S}_1)$ approximate $\mathring{\Psi}$. Then
\begin{equation}
    ||(U_\epsilon(t)-\mathbb{1}) \mathring{\Psi}||_{L^2(\mathcal{S}_1)}\leq ||(U_\epsilon(t)-\mathbb{1}) \mathring{\Psi}_n||_{L^2(\mathcal{S}_1)}+ ||(U_\epsilon(t)-\mathbb{1}) (\mathring{\Psi}-\mathring{\Psi}_n)||_{L^2(\mathcal{S}_1)}\leq ||(U_\epsilon(t)-\mathbb{1}) \mathring{\Psi}_n||_{L^2(\mathcal{S}_1)}+ 2|| \mathring{\Psi}-\mathring{\Psi}_n||_{L^2(\mathcal{S}_1)}.
\end{equation}
Taking a limit of both sides as $t \to 0$ returns
\begin{equation}
    \lim_{t \to 0}||(U_\epsilon(t)-\mathbb{1})\mathring{\Psi}||_{L^2(\mathcal{S}_1)}\leq 2||\mathring{\Psi}-\mathring{\Psi}_n||_{L^2(\mathcal{S}_1)}.
\end{equation}
Since this holds for all $n$, we must have $\lim_{t \to 0}||(U_\epsilon(t)-\mathbb{1})\mathring{\Psi}||_{L^2(\mathcal{S}_1)}=0$ as desired. This concludes our proof that $U_\epsilon(t)$ extends to a $C_0$ contraction semigroup on $L^2(\mathcal{S}_1)$.
\end{proof}
To better understand the evolution mappings $U_\epsilon(t)$, let us recall a well-known result of Phillips and Lumer (originally proven in \cite{Phi59} in the Hilbert space setting, then generalized by \cite{LumPhi61} to Banach space setting) that $C_0$ contraction semigroups are generated by densely defined maximally dissipative operators.
\begin{nonumthm}\label{Lumer-Phillips}
     \cite[Theorem 1.1.3]{Phi59} Let $U(t)$ be a $C_0$ contraction semigroup on $L^2(\mathcal{S}_1)$. Then there exists a linear operator $B$ such that
    \begin{enumerate}
        \item $D(B):= \{ \Psi \in \mathcal{S}_1:\lim_{t \to 0^+}\frac{U(t) \Psi - \Psi}{t} \text{ exists in }L^2(\mathcal{S}_1) \}$ is dense in $L^2(\mathcal{S}_1)\}$
        \item $U(t)=e^{tB}$ i.e.  $\frac{d U(t) \Psi}{dt}\big{|}_{t_0}:=\lim_{t \to t_0}\frac{U(t) \Psi - U(t_0)\Psi}{t-t_0}= B U(t_0) \Psi$ for all $\Psi \in D(B)$, and all time $t_0\geq 0$.
        \item $B$ dissipative, i.e \em$\text{Re}\langle B \Psi, \Psi \rangle_{L^2(\mathcal{S}_1
        )}\leq 0$ \em, and maximal in the class of dissipative operators, i.e if $A$ extends $B$ and $A$ is dissipative then $B=A$.
    \end{enumerate}
    The converse is also true, any densely defined maximally dissipative operator $B$ on $L^2(\mathcal{S}_1)$ generates a $C_0$ contraction semigroup $e^{tB}$.
\end{nonumthm}
\begin{corollary}
    For each $\epsilon>0$, the solution mappings $U_\epsilon(t)$ are generated by some densely defined maximally dissipative operator which we denote $\frac{-i}{\hbar}\hat{H}_\epsilon$, and this operator extends $\frac{-i}{\hbar}\hat{H}$.
\end{corollary}
\begin{proof}
    Let $\mathring{\Psi}\in C_c^\infty(\mathcal{S}_1)=D(\hat{H})$. Since $U_\epsilon(t)\mathring{\Psi}$ is defined to solve (\ref{Leaky IVP}) we have $i\hbar \frac{d}{dt}(U_\epsilon(t) \mathring{\Psi})\big{|}_{t=0}=\hat{H}\Psi$, hence $D(\frac{-i}{\hbar}\hat{H})\subset D(\frac{-i}{\hbar}\hat{H}_\epsilon)$ and $\frac{-i}{\hbar}\hat{H}_\epsilon\mathring{\Psi}=\frac{-i}{\hbar}\hat{H} \mathring{\Psi}$.
\end{proof} Later we will discuss how these maximally dissipative extensions $\frac{-i}{\hbar}\hat{H}_\epsilon$ of $\frac{-i}{\hbar}\hat{H}$ relate to the placement of the leaky boundary conditions on $\mathcal{C}_1 \cup \mathcal{C}_2$. Before this, we will first state the approximation theorem of Kurtz \cite{Kur69}, which gives us a necessary and sufficient condition on the sequence of generators $H_\epsilon$ to guarantee that $U_\epsilon(t)=e^{-\frac{it}{\hbar}\hat{H}_\epsilon}$ converges to another $C_0$ contraction semigroup as $\epsilon \to 0$.

\paragraph{Step 2 in the proof of Thm.~\ref{thm:main}:} We begin by introducing the extended limit operator of the family of generators $\frac{-i}{\hbar}\hat{H}_\epsilon$.
\begin{definition}
    For $\epsilon>0$, let $\frac{-i}{\hbar}\hat{H}_\epsilon$ be a family of densely defined maximally dissipative operators on $L^2(\mathcal{S}_1)$. We denote the \textbf{extended limit} of $\frac{-i}{\hbar}\hat{H}_\epsilon$ as  $\frac{-i}{\hbar}\hat{H}_0$, and define its domain $D(  \frac{-i}{\hbar}\hat{H}_0)$ to be the set of all $\Psi\in L^2(\mathcal{S}_1)$ for which there is a sequence $\Psi_\epsilon \in D(\frac{-i}{\hbar}\hat{H}_\epsilon)$ and an $\eta \in L^2(\mathcal{S}_1)$ such that 
    \begin{equation}
        \lim_{\epsilon \to 0}\Psi_\epsilon=\Psi, \quad \text{and}\quad 
        \lim_{\epsilon \to 0} \frac{-i}{\hbar}\hat{H}_\epsilon \Psi_\epsilon=\eta.
    \end{equation}
    The extended limit operator $\frac{-i}{\hbar}\hat{H}_0$ then acts on elements in its domain according to $\frac{-i}{\hbar}\hat{H}_0\Psi:=\eta$.
\end{definition}
In general the extended limit operator may not be single-valued, although the lemma below gives a sufficient condition for this to be true.
\begin{lemma}\label{Single Valued}\cite[Lemma 1.1]{Kur69}
    If $\frac{-i}{\hbar}\hat{H}_\epsilon$ is a family of maximally dissipative extensions of some densely defined operator $\frac{-i}{\hbar}\hat{H}$, then the extended limit operator $\frac{-i}{\hbar}\hat{H}_0$ is single-valued and also a dissipative extension of $\frac{-i}{\hbar}\hat{H}$.
\end{lemma}
\begin{theorem}\label{Kurtz}\cite[Theorem 2.1]{Kur69} For each $\epsilon>0$, let $U_\epsilon(t)$ be a $C_0$ contraction semigroup on $L^2(\mathcal{S}_1)$ with infinitesimal generators $\frac{-i}{\hbar}\hat{H}_\epsilon$.
    Then there exists a $C_0$ contraction semigroup $U_0(t)$ such that $\lim_{\epsilon \to 0}U_\epsilon(t) \mathring{\Psi}=U_0(t) \mathring{\Psi}$ if and only if $\frac{-i}{\hbar}\hat{H}_0$ is a densely defined maximally dissipative operator. If so then $U_0(t)=e^{-\frac{it}{\hbar}\hat{H}_0}$.
\end{theorem}
Our goal now is to prove that the family of maximally dissipative operators $\frac{-i}{\hbar}\hat{H}_\epsilon$ admit a maximally dissipative extended limit operator. We will do this by writing down an explicit description for all maximally dissipative extensions of $\frac{-i}{\hbar}\hat{H}$. We first note that all dissipative extensions of a densely defined skew-symmetric operator are restrictions of a particular maximal operator.
\begin{definition}
    For a densely defined operator $B$ on $L^2(\mathcal{S}_1)$, we define its adjoint $B^*$
    \begin{equation}
        D(B^*):= \{\Psi \in L^2(\mathcal{S}_1): \text{there exists some } \eta\in L^2(\mathcal{S}_1) \text{ s.t } \langle \Psi,B \Phi \rangle_{L^2(\mathcal{S}_1)}= \langle \eta,\Phi \rangle_{L^2(\mathcal{S}_1)}, \forall \Phi \in D(B) \}, \quad B^*\Psi:=\eta.
    \end{equation}
\end{definition}
\begin{lemma}\cite[Theorem 2.5]{Arendt2023}
    For $\frac{-i}{\hbar}\hat{H}$ a densely defined skew-symmetric operator and $B$ a dissipative operator which extends $\frac{-i}{\hbar}\hat{H}$, we have that $D(B) \subset D(\frac{-i}{\hbar}\hat{H}^*)$ with $B= \frac{-i}{\hbar}\hat{H}^* \big{|}_{D(B)}$.
\end{lemma}
This maximal operator is best understood using a well known decomposition lemma.
\begin{lemma}\cite[Lemma 2.5]{Wegner2017}
    For $\hat{H}$ be a densely defined symmetric operator on $L^2(\mathcal{S}_1)$, the domain of the adjoint can be decomposed as \em
    \begin{equation}
        D(\hat{H}^*)=D(\overline{{H}})\oplus \text{Ker}(i - \hat{H}^*)\oplus \text{Ker}(i+\hat{H}^*).
    \end{equation}
    \em where here $\overline{H}$ denotes the operator closure of $\hat{H}$.
\end{lemma}
We will frequently apply the decomposition lemma to write any element of $D(\hat{H}^*)$ as $\Psi=\Psi_0+\Psi_+ +\Psi_-$ where $\Psi_0 \in D(\overline{H})$, $\hat{H}^*\Psi_\pm=\pm i\Psi_\pm$. Equipping the subspaces $\text{Ker}(i\pm\hat{H}^*)$ with the $L^2(\mathcal{S}_1)$ norm, let $T:\text{Ker}(i+\hat{H}^*)\to \text{Ker}(i-\hat{H}^*)$ be a linear contraction and consider the extension of $\frac{-i}{\hbar}\hat{H}$ defined by
\begin{equation}
    D(\frac{-i}{\hbar}\hat{H}_T):=\left\{\Psi=\Psi_0 + \Psi_+ + \Psi_- \in D(\hat{H}^*): \Psi_+=T\Psi_- \right\}, \quad \frac{-i}{\hbar}\hat{H}_T \Psi:=\frac{-i}{\hbar}\hat{H}^* \Psi.
\end{equation}
All such extensions are dissipative since for any $\Psi \in D(\frac{-i}{\hbar}\hat{H}_T)$
\begin{equation}
    \text{Re}\langle \frac{-i}{\hbar}\hat{H}_T\Psi, \Psi \rangle_{L^2(\mathcal{S}_1)}=\text{Re}\langle \Psi_+ - \Psi_-, \Psi_+ + \Psi_-\rangle_{L^2(\mathcal{S}_1)}=||\Psi_+||_{L^2(\mathcal{S}_1)}- ||\Psi_-||_{L^2(\mathcal{S}_1)}=||T\Psi_-||_{L^2(\mathcal{S}_1)}- ||\Psi_-||_{L^2(\mathcal{S}_1)}\leq 0.
\end{equation}
It is a recent result (originally proven for most cases in \cite{Wegner2017} then generalized to all cases by \cite{Arendt2023}) that all maximally dissipative extensions of the densely defined skew-symmetric operator $\frac{-i}{\hbar}\hat{H}$ are of the form $\frac{-i}{\hbar}\hat{H}_T$ for some linear contraction $T:\text{Ker}(i+\hat{H}^*)\to \text{Ker}(i-\hat{H}^*)$.
\begin{theorem}\label{Arendt}
\cite[Theorem 3.10, Theorem 4.2]{Arendt2023}
    Let $\frac{-i}{\hbar}\hat{H}$ be a densely defined skew-symmetric operator on $L^2(\mathcal{S}_1)$. Then $B$  is a maximally dissipative extension of $\frac{-i}{\hbar}\hat{H}$ if and only if
     there exists a linear contraction \em $T:\text{Ker}(i+\hat{H}^*) \to \text{Ker}(i-\hat{H}^*)$ \em such that $B=\frac{-i}{\hbar}\hat{H}_T$. In addition, $\hat{H}_T$ is self-adjoint if and only if $T$ is unitary. 
\end{theorem}
It follows that associated to our family of generators $\frac{-i}{\hbar}\hat{H}_\epsilon$ we must have some family of linear contractions $T_\epsilon:\text{Ker}(i+\hat{H}^*)\to \text{Ker}(i-\hat{H}^*)$. It is a consequence of the Kurtz approximation theorem that a sufficient condition for $U_\epsilon(t)$ to converge to another $C_0$ contraction semigroup is that the associated sequence of linear contractions $T_\epsilon$ converges to some linear contraction $T_0$.
\begin{proposition}
    Let \em $T_\epsilon:\text{Ker}(i+\hat{H}^*)\to \text{Ker}(i-\hat{H}^*)$ \em be a one parameter family of linear contractions, and $U_\epsilon(t)$ denote the associated $C_0$ contraction semigroups with generator $\frac{-i}{\hbar}\hat{H}_\epsilon$ defined via
    \begin{equation}
        D(\frac{-i}{\hbar}\hat{H}_\epsilon):= \{\Psi \in D(\hat{H}^*): \Psi_+=T_\epsilon \Psi_-\}, \quad \frac{-i}{\hbar}\hat{H}_\epsilon \Psi:=\frac{-i}{\hbar}\hat{H}^*\Psi .
    \end{equation}
    If the linear contractions $T_\epsilon$ converge to another linear contraction $T_0$ as $\epsilon \to 0$, then there exists a $C_0$ contraction semigroup $U_0(t)$ such that $\lim_{\epsilon \to 0}U_\epsilon(t) \mathring{\Psi}=U_0(t) \mathring{\Psi}$.  $U_0(t)$ is generated by the maximally dissipative operator $\frac{-i}{\hbar}\hat{H}_{T_0}$ defined as
    \begin{equation}
        D(\frac{-i}{\hbar}\hat{H}_{T_0}):= \{\Psi \in D(\hat{H}^*): \Psi_+=T_0 \Psi_-\}, \quad \frac{-i}{\hbar}\hat{H}_{T_0} \Psi:=\frac{-i}{\hbar}\hat{H}^*\Psi .
    \end{equation}
\end{proposition}
\begin{proof}
    By Theorem \ref{Kurtz}, it suffices to show that the extended limit operator $\frac{-i}{\hbar}\hat{H}_0$ is equal to $\frac{-i}{\hbar}\hat{H}_{T_0}$. The extended limit operator is single valued and dissipative by Lemma \ref{Single Valued}, and $\frac{-i}{\hbar}\hat{H}_{T_0}$ is maximally dissipative by Theorem \ref{Arendt}, so it suffices to show $ \hat{H}_0$ extends $\hat{H}_{T_0}$ . Let $\Psi= \Psi_0 + T_0\Psi_- + \Psi_- \in D(\hat{H}_{T_0})$, and consider the approximating sequence $\Psi_\epsilon= \Psi_0 + T_\epsilon \Psi_- + \Psi_- \in D(\hat{H}_\epsilon)$. Clearly $\Psi_\epsilon \to \Psi$ as $\epsilon \to 0$, since $T_\epsilon \to T_0$. It is also easy to show that $\hat{H}_\epsilon \Psi_\epsilon \to \hat{H}_{T_0} \Psi$, since
    \begin{equation}
        \hat{H}_\epsilon \Psi_\epsilon= \overline{H}\Psi_0 + \hat{H}^* T_\epsilon \Psi_- + \hat{H}^* \Psi_-=\overline{H}\Psi_0 + i T_\epsilon \Psi_- -i \Psi_- \xrightarrow{\epsilon \to 0} \overline{H}\Psi_0 + i T_0 \Psi_- -i \Psi_-= \hat{H}^* \Psi.
    \end{equation}
    Hence $\Psi \in D(\hat{H}_0)$, and we have $\hat{H}_{T_0} \subset \hat{H}_0$ which implies $\hat{H}_{T_0}=\hat{H}_0$ as desired.
\end{proof}
Our primary goal now is to identify the linear contractions $T_\epsilon$ associated with our leaky boundary evolutions $U_\epsilon(t)$, and prove that these contractions $T_\epsilon$ converge to some $T_0$. It is not obvious how one would do this, since our leaky boundary evolutions came from setting certain boundary conditions along $\mathcal{C}_1$ and $\mathcal{C}_2$, while setting $\Psi_+=T_\epsilon \Psi_-$ is a condition on how $\Psi \in D(\frac{-i}{\hbar}\hat{H}_\epsilon)$ decomposes. It is a remarkable fact that this decomposition condition is equivalent to setting a boundary condition.

\paragraph{Step 3 in the proof of Thm.~\ref{thm:main}:} We begin by computing the deficiency subspaces $\text{Ker}(i \pm \hat{H}^*)$ and show they are uniquely determined by the traces of certain wave function components on $\mathcal{C}_1$ and $\mathcal{C}_2$. For convenience we do this in the case that $m_\text{e}=0$, since there is a one to one correspondence between maximally dissipative extensions of 
$$\frac{-i}{\hbar}\left(\hat{H}- (\mathbb{1}\otimes m_\text{e} \gamma^0 \otimes \mathbb{1})- (\mathbb{1}\otimes \mathbb{1}\otimes m_\text{e} \gamma^0)\right)$$
and maximally dissapative extensions of $\frac{-i}{\hbar}\hat{H}$.
\begin{proposition}
    Let $m_\text{e}=0$, so that $\hat{H}$ acts on wave functions via
    \begin{equation}
        (\hat{H}\Psi)_{\varsigma_0 \varsigma_1 \varsigma_2}=i\hbar (\varsigma_0 \partial_{s_{\text{sph}}}+ \varsigma_1 \partial_{s_{\text{e}_1}}+\varsigma_2 \partial_{s_{\text{e}_2}})\psi_{\varsigma_0 \varsigma_1 \varsigma_2}.
    \end{equation}
    Working in relative coordinates
    \begin{equation}\label{relative coordinates}
        s:=\frac{s_\text{ph}-s_{\text{e}_1}}{2}, \quad \tilde{s}:=\frac{s_{\text{e}_2} - s_\text{ph}}{2}, \quad s_p:=s_\text{ph} 
    \end{equation}
    so that $\mathcal{S}_1=\{ (s_p, s, \tilde{s}) \in \mathbb{R}^3: s>0, \tilde{s}>0)\}$, and the boundaries can be simply written as $\mathcal{C}_1=\{s=0\}$ (parametrized by $s_p,\tilde{s}$) and $\mathcal{C}_2=\{ \tilde{s}=0\}$ (parametrized by $s_p,s$).
    Then the elements of \em $\text{Ker}(i \pm \hat{H}^*)$ \em are fully determined by the (well-defined) traces of their nonzero components along $\mathcal{C}_1$ and $\mathcal{C}_2$, in that
    \begin{eqnarray*}
        \text{Ker}(i - \hat{H}^*) & = & 
        \left\{\Psi \in (L^2(\mathcal{S}_1))^8 \ |\ \Psi(s_p,s,\tilde{s}) = \left(%\begin{pmatrix}
            0,
           0,
            e^{-\frac{s}{\hbar}} g_2(s_p-s,\tilde{s}),
            e^{-\frac{s}{\hbar}} g_3(s_p-s,\tilde{s}+s),
            e^{-\frac{\tilde{s}}{\hbar}} g_4(s_p+\tilde{s},s+\tilde{s}),
            0,
            e^{-\frac{\tilde{s}}{\hbar}} g_6(s_p+\tilde{s},s),
            0\right)^T
        \right\},\\
        \text{Ker}(i + \hat{H}^*) & = & \left\{\Psi \in (L^2(\mathcal{S}_1))^8 \ |\ \Psi(s_p,s,\tilde{s}) = \left(%\begin{pmatrix}
            0,e^{-\frac{\tilde s}{\hbar}} f_1(s_p+\tilde s,s),
            0,
            e^{-\frac{\tilde s}{\hbar}} f_3(s_p+\tilde s,s+\tilde{s}),
            e^{-\frac{{s}}{\hbar}} f_4(s_p-s,\tilde{s}+s),
            e^{-\frac{s}{\hbar}} f_5(s_p-{s},\tilde s ),
            0,
            0\right)^T
        \right\},
    \end{eqnarray*}
    where $f_1,f_3,f_4,f_5$ and $g_2,g_3,g_4,g_6$ are arbitrary functions in $L^2(\mathbb{R}\times\mathbb{R}_+)$.
\end{proposition}
\begin{proof}
    Let $\Psi \in \text{Ker}(i-\hat{H}^*)$. Then the components of $\Psi$ satisfy decoupled first order equations
    \begin{equation}
        \psi_{\varsigma_0 \varsigma_1 \varsigma_2}=\hbar (\varsigma_0 \partial_{s_{\text{sph}}}+ \varsigma_1 \partial_{s_{\text{e}_1}}+\varsigma_2 \partial_{s_{\text{e}_2}})\psi_{\varsigma_0 \varsigma_1 \varsigma_2}.
    \end{equation}
    It is easiest to solve these differential equations in the relative coordinates (\ref{relative coordinates}) using the formulas
    \begin{equation}
        \partial_{s_\text{ph}}=\partial_{s_p}+\frac{1}{2}\partial_s - \frac{1}{2}\partial_{\tilde s}, \quad \partial_{s_{\text{e}_1}}= -\frac{1}{2}\partial_s, \quad \partial_{s_{\text{e}_2}}=\frac{1}{2}\partial_{\tilde s}.
    \end{equation}
    In these coordinates the equation for $\psi_{---}$ reduces to 
    \begin{equation}
        \psi_{---}=-\hbar \partial_{s_p}\psi_{---}. 
    \end{equation}
    The only $L^2$ function which satisfies this equation is $0$, so this component must vanish. 
    \par
    The equation for $\psi_{-+-}$ does admit non-trivial solutions
    \begin{equation}
        \psi_{-+-}=\hbar (-\partial_{s_p}-\partial_s )\psi_{-+-} \quad \Rightarrow \quad \partial_s \psi_{-+-}=\frac{-1}{\hbar}\psi_{-+-} + (-\partial_{s_p})\psi_{-+-}.
    \end{equation}
    It is easy to verify that all solutions to this equation are of the form $\psi_{-+-}=e^{-\frac{s}{\hbar}}(\psi_{-+-})\big{|}_{s=0}(s_p-s, \tilde{s})$, where $(\psi_{-+-})\big{|}_{s=0}\in L^2(\mathbb{R}\times \mathbb{R}_+)$. Repeating these computations returns all solutions for the remaining components.\begin{comment}
    \par
    Just to be sure, let us solve for a component in the case that $\Psi \in \text{Ker}(i+\hat{H}^*)$. The component $\psi_{+--}$ satisfies
    \begin{equation}
        -\psi_{+--}=\hbar (\partial_{s_p}+\partial_s -\partial_{\tilde s})\psi_{+--} \quad \Rightarrow \quad \partial_s \psi_{+--}=\frac{-1}{\hbar}\psi_{+--} + (-\partial_{s_p}+\partial_{\tilde s})\psi_{+--}
    \end{equation}
     so $\psi_{+--}=e^{-\frac{s}{\hbar}}(\psi_{+--})\big{|}_{s=0}(s_p-s, \tilde{s}+s)$, for any function $(\psi_{+--})\big{|}_{s=0}\in L^2(\mathcal{C}_1)$.
     \end{comment}
\end{proof}
\begin{proposition}\label{T sub epsilons}
    Let $m_\text{e}=0$. For $\epsilon \geq 0$, the operators $\frac{-i}{\hbar}\hat{H}_\epsilon$ associated to the leaky boundary conditions
    \begin{equation}
        D(\frac{-i}{\hbar}\hat{H}_\epsilon):=\left\{\Psi \in D(\hat{H}^*): \psi_{-+\varsigma_2} = e^{i\theta_1}\mu_\epsilon \psi_{+-\varsigma_2} \text{ on } \mathcal{C}_1, \quad \psi_{+\varsigma_1-} = e^{i\theta_2}\mu_\epsilon \psi_{-\varsigma_1+}  \text{ on } \mathcal{C}_2 \right\}, \quad \frac{-i}{\hbar}\hat{H}_\epsilon\Psi:= \frac{-i}{\hbar}\hat{H}^*\Psi
    \end{equation}
    are maximally dissipative extensions of $\frac{-i}{\hbar}\hat{H}$, with associated linear contractions \em $T_\epsilon:\text{Ker}(i+\hat{H}^*)\to \text{Ker}(i-\hat{H}^*)$ \em defined by
    \begin{equation}\label{Contraction Def}
    T_\epsilon\begin{pmatrix}
            0
            \\
            e^{-\frac{\tilde s}{\hbar}} f_1(s_p+\tilde s,s)
            \\
            0
            \\
            e^{-\frac{\tilde s}{\hbar}} f_3(s_p+\tilde s,s+\tilde{s})
            \\
            e^{-\frac{{s}}{\hbar}} f_4(s_p-s,\tilde{s}+s)
            \\
            e^{-\frac{s}{\hbar}} f_5(s_p-{s},\tilde s )
            \\
            0
            \\
            0
        \end{pmatrix}:= \begin{pmatrix}
            0
            \\
            0
            \\
            e^{i \theta_1-\frac{s}{\hbar}}\mu_\epsilon(\tilde s)\left( (1-e^{-\frac{2\tilde s}{\hbar}})f_4(s_p-s,\tilde s)+ e^{i \theta_2-\frac{\tilde s}{\hbar}}\mu_\epsilon(\tilde s)f_1(s_p-s+\tilde s,\tilde s)\right)
            \\
            e^{-\frac{s}{\hbar}}\left( \mu_\epsilon(\tilde s+s) e^{i \theta_1}f_5(s_p-s,\tilde{s}+s)-e^{-\frac{\tilde{s}+s}{\hbar}}f_3(s_p+ \tilde{s}-s,\tilde{s}+s)\right)
            \\
            e^{-\frac{\tilde{s}}{\hbar}}\left( \mu_\epsilon(s+\tilde s) e^{i \theta_2}f_1(s_p+\tilde{s},s+\tilde{s})-e^{-\frac{s+\tilde s}{\hbar}}f_4(s_p-s+\tilde s,s+\tilde s)\right)
            \\
            0
            \\
            e^{i \theta_1- \frac{\tilde s}{\hbar}}\mu_\epsilon( s)\left( (1-e^{-\frac{2\tilde s}{\hbar}})f_3(s_p+\tilde s, s)+ e^{i \theta_1-\frac{s}{\hbar}}\mu_\epsilon(s)f_5(s_p+\tilde s-s,s)\right)
            \\
            0
        \end{pmatrix}
\end{equation}
When $\epsilon=0$ the linear mapping $T_0$ is a unitary operator between the deficiency spaces, so $\hat{H}_0$ is a self-adjoint extension of $\hat{H}$. In addition, the linear contractions $T_\epsilon$ converge to $T_0$ as $\epsilon \to 0$, so $\frac{-i}{\hbar}\hat{H}_0$ is the extended limit operator of the family $\frac{-i}{\hbar}\hat{H}_\epsilon$.
\end{proposition}
\begin{proof}
We start by considering wave functions of the form $\Psi_0+\Psi_+ +\Psi_- \in D(\hat{H}^*)$ which satisfy the boundary conditions $\psi_{-+\varsigma_2}= e^{i\theta_1}\mu_\epsilon\psi_{+-\varsigma_2}$ on $\mathcal{C}_1$ and $\psi_{+\varsigma_1-}=e^{i \theta_2}\mu_\epsilon\psi_{-\varsigma_1 +}$ on $\mathcal{C}_2$. We will use these boundary conditions to uniquely write $\Psi_+$ as a function of $\Psi_-$, returning an explicit formula for the mapping $T_\epsilon$. We first note that elements $\Psi_0 \in D(\overline{H})$ already satisfy the boundary conditions, so this term makes no contribution and can be ignored. We begin by writing out any element of the form $\Psi_+ + \Psi_-$
\begin{equation}
    \Psi_+ + \Psi_-= \begin{pmatrix}
            0
            \\
            e^{-\frac{\tilde s}{\hbar}}f_1(s_p+\tilde s, s)
            \\
            e^{-\frac{s}{\hbar}} g_2(s_p-s,\tilde{s})
            \\
            e^{-\frac{s}{\hbar}} g_3(s_p-s,\tilde{s}+s)+e^{\frac{-\tilde s}{\hbar}}f_3(s_p+\tilde s, s+\tilde s)
            \\
            e^{-\frac{\tilde{s}}{\hbar}} g_4(s_p+\tilde{s},s+\tilde{s})+e^{-\frac{s}{\hbar}}f_4(s_p-s,\tilde s +s)
            \\
            e^{-\frac{s}{\hbar}}f_5(s_p-s,\tilde s)
            \\
            e^{-\frac{\tilde{s}}{\hbar}} g_6(s_p+\tilde{s},s)
            \\
            0
        \end{pmatrix}
\end{equation}
Requiring that such elements satisfy $\psi_{+--}=e^{i \theta_2}\mu_\epsilon(s) \psi_{--+}$ on $\mathcal{C}_2$ and $\psi_{-+-}=e^{i \theta_1}\mu_\epsilon(\tilde s) \psi_{+--}$ on $\mathcal{C}_1$ returns explicit relations for $g_4$ and $g_2$ in terms of $f_1$ and $f_4$. 
\begin{equation}
    (\Psi_+ + \Psi_-)_{+--}\Big{|}_{\mathcal{C}_2}=e^{i \theta_2}\mu_\epsilon(s)(\Psi_+ + \Psi_-)_{--+}\Big{|}_{\mathcal{C}_2} \quad \iff \quad g_4(s_p,s)=e^{i \theta_2}\mu_\epsilon(s)f_1(s_p,s) - e^{-\frac{ s}{\hbar}}f_4(s_p-s,s)
\end{equation}
\begin{equation}
    (\Psi_+ + \Psi_-)_{-+-}\Big{|}_{\mathcal{C}_1}=e^{i \theta_1}\mu_\epsilon(\tilde s)(\Psi_+ + \Psi_-)_{+--}\Big{|}_{\mathcal{C}_1} \quad \iff \quad g_2(s_p,\tilde s)=e^{i \theta_1}\mu_\epsilon(\tilde s)\left( f_4(s_p,\tilde s)+ e^{-\frac{\tilde s}{\hbar}}g_4(s_p+\tilde s, \tilde s)\right)
\end{equation}
\begin{comment}
Hence
\begin{equation}
    \begin{pmatrix}
        g_4(s_p,s)=e^{i \theta_2}\mu_\epsilon(s)f_1(s_p,s) - e^{-\frac{\tilde s}{\hbar}}f_4(s_p-s,s) \text{ on } \mathcal{C}_2
        \\
        g_2(s_p,\tilde s)=e^{i \theta_1}\mu_\epsilon(\tilde s)\left( f_4(s_p,\tilde s)+ e^{-\frac{\tilde s}{\hbar}}\left(e^{i \theta_2}\mu_\epsilon(\tilde s)f_1(s_p+\tilde s,\tilde s)- e^{-\frac{\tilde s}{\hbar}}f_4(s_p,\tilde s)\right)\right) \text{ on }\mathcal{C}_1
    \end{pmatrix}
\end{equation}
\end{comment}
The other set of boundary conditions returns
\begin{equation}
    (\Psi_+ + \Psi_-)_{-++}\Big{|}_{\mathcal{C}_1}=e^{i \theta_1}\mu_\epsilon(\tilde s)(\Psi_+ + \Psi_-)_{+-+}\Big{|}_{\mathcal{C}_1} \quad \iff \quad g_3(s_p,\tilde s)=e^{i \theta_1}\mu_\epsilon(\tilde s)f_5(s_p,\tilde s) - e^{-\frac{\tilde s}{\hbar}}f_3(s_p+\tilde s,\tilde s)
\end{equation}
\begin{equation}
    (\Psi_+ + \Psi_-)_{++-}\Big{|}_{\mathcal{C}_2}=e^{i \theta_2}\mu_\epsilon(s)(\Psi_+ + \Psi_-)_{-++}\Big{|}_{\mathcal{C}_2} \quad \iff \quad g_6(s_p, s)=e^{i \theta_2}\mu_\epsilon( s)\left( f_3(s_p, s)+ e^{-\frac{s}{\hbar}}g_3(s_p- s,  s)\right)
\end{equation}
\begin{comment} Hence
\begin{equation}
    \begin{pmatrix}
        g_3(s_p,\tilde s)=e^{i \theta_1}\mu_\epsilon(\tilde s)f_5(s_p,\tilde s) - e^{-\frac{\tilde s}{\hbar}}f_3(s_p+\tilde s,\tilde s) \text{ on } \mathcal{C}_2
        \\
        g_6(s_p, s)=e^{i \theta_2}\mu_\epsilon( s)\left( f_3(s_p, s)+ e^{-\frac{s}{\hbar}}\left(e^{i \theta_1}\mu_\epsilon(s)f_5(s_p-s,s)-e^{-\frac{s}{\hbar}}f_3(s_p,s) \right)\right) \text{ on }\mathcal{C}_1
    \end{pmatrix}
\end{equation}
\end{comment}
Since $\Psi_+$ is uniquely defined by $g_2,g_3,g_4$ and $g_6$, the relations above define an expression for $T_\epsilon$. To prove these are contractions we first recall that
\begin{equation}
\begin{split}
    \langle -i\hat{H}^*(T_\epsilon\Psi_-+\Psi_-), T_\epsilon\Psi_- + \Psi_-\rangle_{L^2(\mathcal{S}_1)}&=\langle T_\epsilon \Psi_- - \Psi_-, T_\epsilon \Psi_- +\Psi_-\rangle_{L^2(\mathcal{S}_1)}
    \\
    &=||T_\epsilon \Psi_-||^2_{L^2(\mathcal{S}_1)}- || \Psi_-||^2_{L^2(\mathcal{S}_1)}
\end{split}
\end{equation}
By the continuity equation, we can evaluate the L.H.S as 
\begin{equation}
    \langle -i\hat{H}^*(T_\epsilon\Psi_-+\Psi_-), T_\epsilon\Psi_- + \Psi_-\rangle_{L^2(\mathcal{S}_1)}+{\langle  T_\epsilon\Psi_- + \Psi_-, -i\hat{H}^*(T_\epsilon\Psi_-+\Psi_-)\rangle_{L^2(\mathcal{S}_1)}}= \int_{\mathcal{C}_1 \cup \mathcal{C}_2}\vec{j}_{(T_\epsilon \Psi_- + \Psi_-)}\cdot \vec{n}dS
\end{equation}
where $\vec{j}_{T_\epsilon \Psi_- + \Psi_-}$ denotes the probability current associated to the wave function $T_\epsilon \Psi_- + \Psi_-$. By virtue of the boundary conditions, this flux is non-positive for $\epsilon>0$ and vanishes when $\epsilon=0$, hence $T_\epsilon$ is a linear contraction for $\epsilon>0$ and norm preserving when $\epsilon=0$. $T_0$ is also invertible, which can be seen by inverting the boundary conditions and solving for $f_1, f_3,f_4,f_5$ in terms of $g_2,g_3,g_4,g_6$. Hence $T_0$ is unitary.
\par
Convergence of the linear contractions $T_\epsilon$ follows trivially from the $L^\infty$ convergence of the transition functions $\mu_\epsilon$ to the constant function $\mu_0=1$ as $\epsilon \to 0$.
\end{proof}
\begin{proposition}\label{Final Prop}
    For all $m_\text{e} \geq 0$ and $\epsilon > 0$, the operators associated to the leaky boundary conditions 
    \begin{equation}
        D(\frac{-i}{\hbar}\hat{H}_\epsilon):=\left\{\Psi \in D(\hat{H}^*): \psi_{-+\varsigma_2} = e^{i\theta_1}\mu_\epsilon \psi_{+-\varsigma_2} \text{ on } \mathcal{C}_1, \quad \psi_{+\varsigma_1-} = e^{i\theta_2}\mu_\epsilon \psi_{-\varsigma_1+}  \text{ on } \mathcal{C}_2 \right\}, \quad \frac{-i}{\hbar}\hat{H}_\epsilon\Psi:= \frac{-i}{\hbar}\hat{H}^*\Psi
    \end{equation}
    are maximally dissipative extensions of $\frac{-i}{\hbar}\hat{H}$, and the extended limit operator of these extensions is $\frac{-i}{\hbar}\tilde{H}$, where $\tilde{H}$ is the self adjoint extension of $\hat{H}$ with domain \begin{equation}
    D(\tilde{H}):=\left\{\Psi \in D(\hat{H}^*): \psi_{-+\varsigma_2} = e^{i\theta_1}\psi_{+-\varsigma_2} \text{ on } \mathcal{C}_1, \quad \psi_{+\varsigma_1-} = e^{i\theta_2}\psi_{-\varsigma_1+}  \text{ on } \mathcal{C}_2 \right\}.
\end{equation}
\end{proposition}
\begin{proof}
    For all $m_\text{e}\geq 0$, the operator $M:=\mathbb{1}\otimes m_\text{e}\gamma^0 \otimes \mathbb{1}+\mathbb{1}\otimes \mathbb{1}\otimes m_\text{e}\gamma^0$ is bounded and symmetric on $L^2(\mathcal{S}_1,\mathbb{C}^8)$, so $(\hat{H}-M)^*=\hat{H}^*-M$ and $D(\frac{-i}{\hbar}(\hat{H}_\epsilon - M))=D(\frac{-i}{\hbar}\hat{H}_\epsilon)$. By Proposition \ref{T sub epsilons} we have that $\frac{-i}{\hbar}(\hat{H}_\epsilon-M)$ are maximally dissipative extensions of $\frac{-i}{\hbar}(\hat{H}-M)$, hence $\frac{-i}{\hbar}\hat{H}_\epsilon$ are maximally dissipative extensions of $\frac{-i}{\hbar}\tilde{H}$. Since the extended limit operator of $\frac{-i}{\hbar}(\hat{H}_\epsilon-M)$ is $\frac{-i}{\hbar}(\hat{H}-M)$, it is easy to verify that $\frac{-i}{\hbar}\tilde{H}$ is the extended limit operator of $\frac{-i}{\hbar}\hat{H}_\epsilon$. Lastly, $\tilde{H}$ is self-adjoint since $\tilde{H}-M$ is self-adjoint.
\end{proof}
Theorem \ref{thm:main} then immediately follows from Theorem \ref{Kurtz} and Proposition \ref{Final Prop}.
\subsubsection{Extension to Multi-Time Setting}
In the previous section, we proved convergence of the infinite diagram sum depicted in figure \ref{fig:Infinite Diagram Sum} in the equal-time setting. We did this by considering an altered evolution of our wave function which was equivalent to discarding all diagrams in which the photon bounced between the two electrons while they were "too close" (i.e when $|s_{\text{e}_2}-s_{\text{e}_1}|<\epsilon$), and then took a limit as our alteration went to zero. 
\par
A natural way to extend our method to the multi-time setting would be to discard all diagrams in which the photon bounced between the two electrons while $-\eta(\textbf{x}_{\text{e}_1}-\textbf{x}_{\text{e}_2},\textbf{x}_{\text{e}_1}-\textbf{x}_{\text{e}_2})<\epsilon$. This corresponds to the following multi-time IBVP
\begin{equation}\label{Leaky Multi-Time}
\left\{
\begin{array}{rclr} 
-i\hbar D_{\text{ph}} \Psi &= &0  & 
\\
-i\hbar D_{\text{e}_1} \Psi + m_{e} \Psi &= &0 & 
\\ 
-i\hbar D_{\text{e}_2} \Psi + m_\text{e} \Psi &= &0  &\text{in } \mathcal{S}_1 
\\
\Psi &= &\mathring \Psi  &\text{on } \mathcal{I} 
\\
\psi_{-+\varsigma_2} &= &e^{i\theta_1}\mu_\epsilon \psi_{+-\varsigma_2} \hspace{0.5cm}&\text{on } \mathcal{C}_1
\\
\psi_{+\varsigma_1-} &= &e^{i\theta_2} \mu_\epsilon \psi_{-\varsigma_1+}  &\text{on } \mathcal{C}_2
\end{array}
\right.
\end{equation}
where here, $\mu_\epsilon=\mu_\epsilon(-\eta(\textbf{x}_{\text{e}_1}-\textbf{x}_{\text{e}_2},\textbf{x}_{\text{e}_1}-\textbf{x}_{\text{e}_2}))$ is a smooth transition function. Unfortunately, the multi-time IBVP above is not well-posed. To see this, consider the dynamics of the wave function on $\mathcal{C}_1$. Both $\psi_{-+\varsigma_2}$ and $\psi_{+-\varsigma_2}$ satisfy the electron $2$ evolution equation on $\mathcal{C}_1$
\begin{equation}
    (\partial_{t_{\text{e}_2}}-\varsigma_2 \partial_{s_{\text{e}_2}}) \psi_{\varsigma_0 \varsigma_1 \varsigma_2} + i\omega \psi_{\varsigma_0 \varsigma_1 \overline{\varsigma}_2}=0.
\end{equation}
For this to be consistent with the boundary condition on $\mathcal{C}_1$, we would need 
\begin{equation}
    (\partial_{t_{\text{e}_2}}-\varsigma_2 \partial_{s_{\text{e}_2}}) \mu_\epsilon=0
\end{equation}
which cannot hold for non-constant functions of $\eta(\textbf{x}_{\text{e}_1}-\textbf{x}_{\text{e}_2},\textbf{x}_{\text{e}_1}-\textbf{x}_{\text{e}_2})$. This makes it impossible to generalize the equal-time "leaky boundary" problem to the multi-time setting.
\par
Luckily, the equal-time evolution is enough for our purposes, so long as we combine it with our previous multi-time results. Let $(\textbf{x}_\text{ph},\textbf{x}_{\text{e}_1},\textbf{x}_{\text{e}_2})$ be a Coulomb configuration. We wish to write down the solution for $\Psi$ satisfying the original multi-time IBVP (\ref{three body problem}) at this configuration in terms of the initial data $\mathring{\Psi}$. First, let $t=\min (t_\text{ph},t_{\text{e}_1},t_{\text{e}_2})$. We may evolve the initial data up to the equal-time hypersurface in $\mathcal{S}_1$ written as $\{t_\text{ph}=t_{\text{e}_1}=t_{\text{e}_2}=t \}$ using the equal-time evolution operator. Let $\Psi \big{|}_t =e^{\frac{-it}{\hbar}\hat{H}_\theta}\mathring{\Psi}$. To solve for $\Psi$ at the Coulomb configuration, we only need to evolve our new initial data $\Psi \big{|}_t$ further in only two of the time variables. For $\textbf{x}\in \mathbb{R}^{1,1}$, let $\textbf{x}^t:=(x^0-t,x^1)$. There are only three possible cases:
\newline
\textbf{Free Evolution:} If $(\textbf{x}^t_\text{ph},\textbf{x}^t_{\text{e}_1},\textbf{x}^t_{\text{e}_2})\in \mathcal{F}$, then
\begin{equation}
    \Psi(\textbf{x}_\text{ph},\textbf{x}_{\text{e}_1},\textbf{x}_{\text{e}_2})= \big(\mathcal{P}(\textbf{x}^t_\text{ph})\otimes \mathcal{E}(\textbf{x}^t_{\text{e}_1}) \otimes  \mathcal{E}(\textbf{x}^t_{\text{e}_2}) \big) \Psi \big{|}_t.
\end{equation}
\textbf{One more bounce:} If $(\textbf{x}^t_\text{ph},\textbf{x}^t_{\text{e}_1},\textbf{x}^t_{\text{e}_2})\in \mathcal{X}$, then either $(\textbf{x}^t_\text{ph},\textbf{x}^t_{\text{e}_1}) \in \mathcal{N}$ or $(\textbf{x}^t_\text{ph},\textbf{x}^t_{\text{e}_2}) \in \mathcal{N}$. Suppose the former. Then
\begin{equation}
    \Psi(\textbf{x}_\text{ph},\textbf{x}_{\text{e}_1},\textbf{x}_{\text{e}_2})=\big(C_{\theta_1}(\textbf{x}^t_\text{ph},\textbf{x}^t_{\text{e}_1}) \otimes  \mathcal{E}(\textbf{x}^t_{\text{e}_2}) \big) \Psi \big{|}_t
\end{equation}
while an analogous equation holds for the other case. This completes our construction of the solution to the initial-boundary value problem (\ref{three body problem}). 
\section{Summary and Outlook}
Using a relativistic quantum-\textit{mechanical} treatment of the Compton effect in $1+1$ dimensions, we have shown that placing a photon between two identical massive spin-half Dirac particles generates an effective interaction which can be decomposed into many individual two-particle interactions. 

Our next goal is to further study this effective interaction, and to tease out the connection that we believe is there between the arbitrary phase shifts $\theta_1,\theta_2$ and an effective ``charge" for the massive particles in this model.

Ultimately, our goal is to find out whether interactions between charged particles can be reproduced in terms of relativistic quantum mechanics of a \textit{fixed} finite number of photons and  spin-half particles, at least in $1+1$ dimensions. Extending our work to the physical $3+1$ dimensional setting in a way that accurately reproduces empirical electromagnetism would require us to overcome many challenges, the most significant of which is that in more than one space dimension, the coincidence hypersurfaces will have a higher co-dimension in the multi-time multi-particle configuration space.  We have reasons to be optimistic however that a geometric approach to the problem could provide a way out of this dilemma. 
\section{Acknowledgment} The authors are grateful to the anonymous referees for extremely helpful suggestions.
\section*{\centering APPENDIX}
\appendix 
\section{Solution formulas for the Klein-Gordon and Dirac equation}
The initial value problem for the Dirac equation is given by
\begin{equation}
    \left\{\begin{array}{rcl}
        -i\hbar D_\text{e} \Psi_\text{e} + m_\text{e} \Psi_\text{e} &=&0
        \\
        \Psi_\text{e} \big{|}_{t_\text{e}=0}&=&\mathring{\Psi}_\text{e}
    \end{array}\right..
    \end{equation}
    where $\Psi_\text{e}=\begin{pmatrix}
        \psi_-
        \\
        \psi_+
    \end{pmatrix}$. To solve the Dirac IVP, we decouple the components of $\Psi_\text{e}$ and write down the second order initial value problems
    \begin{equation}\label{KG IVP}
       \left\{ \begin{array}{rcl}
            \Box_\text{e} \psi_\varsigma+\omega^2 \psi_\varsigma&=&0
            \\
            \psi_\varsigma \big{|}_{t_\text{e}=0}&=&\mathring{\psi}_\varsigma(s_\text{e})
            \\
            \partial_{t_\text{e}} \psi_\varsigma\big{|}_{t_\text{e}=0}&=&\varsigma \partial_{s_{\text{e}}}\mathring{\psi}_\varsigma(s_{\text{e}})-i\omega\mathring{\psi}_{\overline{\varsigma}}(s_{\text{e}})
        \end{array}\right..
    \end{equation}
    Recall that the general solution to the Cauchy problem for the Klein-Gordon Equation
\begin{equation}\label{cauchy problem original}
w_{tt}-w_{ss}+\omega^2 w = 0,
\hspace{0.2cm}
w(0,s) = A(s),
\hspace{0.2cm}
w_{t}(0,s) = B(s),
\end{equation}
is
\begin{equation} 
\begin{split}
w(t,s) = \frac{1}{2} \{ A(s-t) + A(s+t) \} - &\frac{\omega t}{2} \int_{s-t}^{s+t} \frac{J_{1}(\omega\sqrt{t^2-(s-\sigma)^2})}{\sqrt{t^2-(s-\sigma)^2}} A(\sigma) d\sigma \\
+ &\frac{1}{2} \int_{s-t}^{s+t} J_{0}(\omega\sqrt{t^2-(s-\sigma)^2}) B(\sigma) d\sigma
\end{split}
\end{equation}
where $J_{\nu}$ is the Bessel function of index $\nu$. Applying this to the Dirac equation and  using integration by parts, our Dirac evolution can be written as
\begin{equation}\label{Dirac Solution}
\begin{split}
\psi_{\varsigma}(t_\text{e},s_\text{e})&=\mathring{\psi}_{\varsigma}(s_{\text{e}}+\varsigma t_{\text{e}})
\\
&-\frac{\omega}{2}\int_{s_{\text{e}}-t_{\text{e}}}^{s_{\text{e}}+t_{\text{e}}} J_{1}(\omega\sqrt{t_{\text{e}}^2-(s_{\text{e}}-\sigma)^2})\frac{\sqrt{t_{\text{e}}-\varsigma(s_{\text{e}}-\sigma)}}{\sqrt{t_{\text{e}}+\varsigma(s_{\text{e}}-\sigma)}}\mathring{\psi}_{\varsigma}(\sigma) d\sigma
\\
&-\frac{i\omega}{2}\int_{s_{\text{e}}-t_{\text{e}}}^{s_{\text{e}}+t_{\text{e}}} J_{0}(\omega\sqrt{t_{\text{e}}^2-(s_{\text{e}}-\sigma)^2}) \mathring{\psi}_{\overline{\varsigma}}(\sigma) d\sigma
\end{split}
\end{equation}
We condense such equations by introducing the electronic time evolution operator $\mathcal{E}(t_\text{e},s_\text{e})$ defined by $\Psi_\text{e}(t_\text{e},s_\text{e})=\mathcal{E}(t_\text{e},s_\text{e})\mathring{\Psi}_\text{e}$. 
% Since $\frac{J_1(\omega x)}{x}\leq \frac{\omega}{2}$, and $J_0(x)\leq 1$, we have the following small-time growth estimate
% \begin{equation}\label{Dirac Estimate}
%     ||\mathcal{E}(t_\text{e},\cdot)\mathring{\Psi}||_\infty\leq (1+\omega^2t_\text{e} +\omega t_\text{e})||\mathring{\Psi}||_\infty
% \end{equation}
\par
Also relevant is the Goursat problem for the 1-dimensional Klein-Gordon equation,
\begin{equation}
\left\{\begin{array}{rcl}
             U_{tt}-U_{ss}+\omega^2 U&=&0 \quad \text{for } |s|<t
            \\
            U(b,b)&=&F(b)
            \\
            U(c,-c) &=&G(c)
        \end{array}\right.,
\end{equation}
which has the solution 
\begin{equation}\label{goursat solution}
\begin{split}
U(t,s) = &F(\frac{t+s}{2}) + G(\frac{t-s}{2}) - \frac{1}{2}(F(0) + G(0))J_{0}(\omega\sqrt{t^2-s^2})
\\
&-\omega(t-s)\int_{0}^{\frac{t+s}{2}}F(b)\frac{J_{1}(\omega\sqrt{(t-s)(t+s-2b)})}{\sqrt{(t-s)(t+s-2b)}}db
\\
&-\omega(t+s)\int_{0}^{\frac{t-s}{2}}G(c)\frac{J_{1}(\omega\sqrt{(t+s)(t-s-2c)})}{\sqrt{(t+s)(t-s-2c)}}dc.
\end{split}
\end{equation}
By linearity, we may introduce Goursat evolution operators $G^R(t,s)$ and $G^L(t,s)$, defined so that
\begin{equation}
    U(t,s)=G^R(t,s)F + G^L(t,s)G.
\end{equation}
Explicitly, these operators are written as
\begin{equation}\label{GR Def}
G^R(t,s)F = F(\frac{t+s}{2}) - \frac{F(0)}{2} J_{0}(\omega\sqrt{t^2-s^2})
-\omega(t-s)\int_{0}^{\frac{t+s}{2}}F(b)\frac{J_{1}(\omega\sqrt{(t-s)(t+s-2b)})}{\sqrt{(t-s)(t+s-2b)}}db
\end{equation}
\begin{equation}\label{GL Def}
G^L(t,s)G = G(\frac{t-s}{2}) - \frac{G(0)}{2} J_{0}(\omega\sqrt{t^2-s^2})-\omega(t+s)\int_{0}^{\frac{t-s}{2}}G(c)\frac{J_{1}(\omega\sqrt{(t+s)(t-s-2c)})}{\sqrt{(t+s)(t-s-2c)}}dc.
\end{equation}

\section{Well-Posedness by Fixed-Point Argument}\label{AppendixLN}
In this section we outline the proof of \cite[Theorem 4.3]{LiNi2020} for the special case of two particles, one photon and one electron.  The proof relies on a fixed-point argument. For the photon-electron system, the fixed-point argument may not seem as valuable as the proof of the same result in \cite[Theorem 5.1]{KLTZ}, because it does not construct an explicit algorithm for the time evolution of $\Psi$. But the fixed-point-based proof is easily generalized to the many-particle case, as was done in full generality in \cite{LiNi2020}. We will now state the theorem.
\begin{theorem}\label{thm:2body}
Let $\theta \in [0,2\pi)$ be a constant phase. Let $\mathcal{S}_1$ denote the set of space-like configurations in $\mathcal{M}$, as was defined in (\ref{2body S1 definitionL}), and let $\mathcal{C} \subset \partial \mathcal{S}_1$ be the coincidence set defined by $\{ \textbf{x}_\text{ph}=\textbf{x}_\text{e}\}$, and let $\mathcal{I}$ be the initial surface. Let  the function $\mathring \Psi : \mathcal{I} \rightarrow \mathbb{C}^4$ be $C^1$ and compactly supported in the half-space $\mathcal{I} \cap \overline{\mathcal{S}_1}$, and suppose that the initial-value function also satisfies the boundary conditions given by (\ref{2 body psi boundaryL}). Then the following initial-boundary value problem for the wave function $\Psi: \mathcal{M} \rightarrow \mathbb{C}^4$,
\begin{equation}
\left\{
\begin{array}{rclr} 
-i\hbar D_{\text{ph}} \Psi &= &0  & 
\\
-i\hbar D_{\text{e}} \Psi + m_\text{e} \Psi &= &0  &\text{in } \mathcal{S}_1 
\\
\Psi &= &\mathring \Psi  &\text{on } \mathcal{I} 
\\
\psi_{+-} &= &e^{i\theta}\psi_{-+} \hspace{0.5cm}&\text{on } \mathcal{C}
\end{array}
\right.
\end{equation}
has a unique global-in-time solution that is supported in $\overline{\mathcal{S}_1}$, and depends continuously on the initial data $\mathring \Psi$.
\end{theorem}
\begin{proof}
    We specialize the proof of \cite{LiNi2020}, who established the result for a much more general class of IBVPs, one that  allowed data to transfer from the $N$ particle sector to other sectors. Solving the initial-boundary value problem above is equivalent to solving the following system of integral equations
\begin{equation}\label{FP1}
    \psi_{-\varsigma_1}(\textbf{x}_\text{ph},\textbf{x}_\text{e})= \mathring{\psi}_{-\varsigma_1}(s_\text{ph}-t_\text{ph},s_\text{e}+\varsigma_2 t) - i\omega \int_{0}^{t_\text{e}}\psi_{-\overline{\varsigma}_1}(0,s_\text{ph}-t_\text{ph},\tau, s_\text{e}+\varsigma_1 (t-\tau))d\tau
\end{equation}
\begin{equation}\label{FP2}
    \psi_{++}(\textbf{x}_\text{ph},\textbf{x}_\text{e})= \mathring{\psi}_{++}(s_\text{ph}+t_\text{ph},s_\text{e}+ t) - i\omega \int_{0}^{t_\text{e}}\psi_{+-}(0,s_\text{ph}+t_\text{ph},\tau, s_\text{e}+ (t-\tau))d\tau
\end{equation}
\begin{equation}\label{FP3}
    \psi_{+-}(\textbf{x}_\text{ph},\textbf{x}_\text{e})= \left\{\begin{array}{ll}
    \mathring{\psi}_{+-}(s_\text{ph}+t_\text{ph},s_\text{e}- t) - i\omega \int_{0}^{t_\text{e}}\psi_{++}(\textbf{x}_\text{ph},\tau, s_\text{e}- (t-\tau))d\tau \quad &\mbox{ on }\mathcal{F},
    \\ 
    e^{i\theta} \left[\mathring{\psi}_{-+}(s_\text{e}-t_\text{e},s_\text{ph}+t_\text{ph})-i\omega \int_{0}^{T} \psi_{--}(T,S,\tau, s_\text{ph}+ (t_\text{ph}-\tau))d\tau \right] \\
    \hspace{1in}\mbox{}-i\omega\int_{T}^{t_\text{e}}\psi_{++}(\textbf{x}_\text{ph},\tau, s_\text{e}- (t-\tau))d\tau \quad &\mbox{ on }\mathcal{N}.
    \end{array}\right.
\end{equation}
Here
\begin{equation}
    T:=\frac{1}{2}(s_\text{ph}+t_\text{ph}-s_\text{e}+t_\text{e})>0, \quad S:=\frac{1}{2}(s_\text{ph}+t_\text{ph}+s_\text{e}-t_\text{e})
\end{equation}
is the space-time location where the backwards light-cones of the photon and electron cross. To prove existence of a solution to the integral equations we use the standard fixed-point method. Let
\begin{equation}
    \mathcal{S}^T:=\left\{ (\textbf{x}_\text{ph}, \textbf{x}_{\text{e}})\in\mathcal{S}_1 | t_\text{ph},t_\text{e} \in [0,T]\right\}
\end{equation}
and
\begin{equation}
    \mathcal{D}:=\left\{\Psi \in C^1_b(\mathcal{S}^T,\mathbb{C}^4)| \Psi=\mathring{\Psi} \text{ on } \mathcal{I}\right\}.
\end{equation}
Let $\gamma>0$. We equip the function space $\mathcal{D}$ with the weighted norm
\begin{equation}
    ||\Psi||_\gamma := ||e^{\gamma t_\text{e}}\Psi||_{L^\infty(\mathcal{S}^T)} + \max_{\mu \in \{t_\text{ph},s_\text{ph},t_\text{e},s_\text{e} \}}||e^{\gamma t_\text{e}}\partial_\mu \Psi||_{L^\infty(\mathcal{S}^T)}.
\end{equation}
Clearly $\mathcal{D}$ equipped with this norm is a Banach space. Let $F$ denote the mapping defined by equations (\ref{FP1},\ref{FP2},\ref{FP3}), so that these equations reduce to $\Psi=F(\Psi)$. It suffices to show that $F$ is a well-defined map from $\mathcal{D}\to \mathcal{D}$, and that 
\begin{equation}
        ||F(\Psi)-F(\Phi)||_\gamma \leq \frac{4 \omega}{\gamma} ||\Psi-\Phi||_\gamma.
    \end{equation}
Setting $\gamma>4\omega$, it follows that $F$
contraction mapping on $\mathcal{D}$, so that by Banach's fixed point theorem there exists a unique wave function fixed by the integral mapping.
\end{proof}
This technique is applied in \cite{LiNi2020} to systems with many massive particles, and so also holds for the case of a 3-body wave function of one photon between two electrons.
\section{Data Availability and Conflict of Interest Statement} This manuscript does not use or contain any data.  On behalf of all authors, the corresponding author states that there is no conflict of interest.

%\printbibliography
\bibliographystyle{plain}
%\bibliography{mybibliography}

\end{document}